\documentclass{article}
\usepackage{euscript,amsmath,amssymb,amsfonts,graphicx,bm,amsthm,mathtools}
\usepackage{cite}

  \usepackage{paralist}
  \usepackage{graphics} 
\usepackage[caption=false]{subfig}
\usepackage{soul}
\usepackage{bbm}

\usepackage{latexsym,amssymb,enumerate,amsmath,amsthm} 
  \usepackage{graphicx} 
  \usepackage{tikz}
     \usepackage[colorlinks=false]{hyperref}
 \hypersetup{urlcolor=blue, citecolor=red}
\usepackage{latexsym,amssymb,enumerate,amsmath}

\usepackage{pgfplots}
\usepackage{soul,xcolor}

\newcommand{\dd}{\textup{d}}
\def\eps{\varepsilon}
\def\E{\mathbb{E}}
\def\P{\mathbb{P}}
\def\R{\mathbb{R}}

\def\t{\textup{target}}

\def\D{\prescript{}{0}D_{t}^{1-\alpha}}

\newtheorem{theorem}{Theorem}
\newtheorem{lemma}{Lemma}
\newtheorem{proposition}{Proposition}

\theoremstyle{plain}
\theoremstyle{remark}
\newtheorem{remark}{Remark}

\begin{document}


\title{First passage times under\\ frequent stochastic resetting}


\author{Samantha Linn\thanks{Department of Mathematics, University of Utah, Salt Lake City, UT 84112 USA (\texttt{samantha.linn@utah.edu}). SL was supported by the National Science Foundation (Grant No. 2139322).} \and Sean D. Lawley\thanks{Department of Mathematics, University of Utah, Salt Lake City, UT 84112 USA (\texttt{lawley@math.utah.edu}). SDL was supported by the National Science Foundation (Grant Nos.\ CAREER DMS-1944574 and DMS-1814832).}
}
\date{\today}
\maketitle

\begin{abstract}
We determine the full distribution and moments of the first passage time for a wide class of stochastic search processes in the limit of frequent stochastic resetting. Our results apply to any system whose short-time behavior of the search process without resetting can be estimated. In addition to the typical case of exponentially distributed resetting times, we prove our results for a wide array of resetting time distributions. We illustrate our results in several examples and show that the errors of our approximations vanish exponentially fast in typical scenarios of diffusive search.
\end{abstract}

\section{Introduction}
For many decades, there has been sustained interest in understanding first passage times (FPTs), which characterize the time it takes a searcher to find a target \cite{Siegert1951, Slepian1961, Szabo1980, redner2001, benichou2011, kurella2015}. Search processes of interest include pure diffusion \cite{Roberts1986, lindsay2017}, anomalous diffusion \cite{gitterman2000, palyulin2019}, random walks on discrete networks \cite{noh2004, masuda2017}, run and tumble particles \cite{evans2018, santra2020}, inactivating searchers \cite{grebenkov2017, ma2020}, and so on.

More recently, there has been a strong interest in FPTs of searchers under stochastic resetting, which means that the searcher is reset to elsewhere in the state space at random times \cite{evans2020, aPal_2015_PRE, masoliver_2019_PRE, Bressloff_2020_jphysA2}. Biological approximations of such behavior vary widely in spatial and temporal scale, for example, RNA cleavage during transcription \cite{Roldan_16} and predator dynamics during foraging \cite{Mercado_2018}. In theoretical treatments of these systems, it has been shown that stochastic resetting can reduce the expected search time \cite{Schumm_2021}. For instance, consider a diffusing searcher on the half line with an absorbing boundary condition at the origin. It is well-known that the mean FPT to the origin is infinite. However, if the searcher undergoes resetting according to an exponential random variable with rate $r\in(0,\infty)$, then the mean FPT is finite \cite{evans2020}. This result can be heuristically explained by the fact that the searcher resets its position before wandering too far from the origin (see Figure \ref{figschem}a for an illustration).

\begin{figure}
  \centering
             {\scriptsize(a)}\hspace{-.05\textwidth}
             \includegraphics[width=0.495\textwidth]{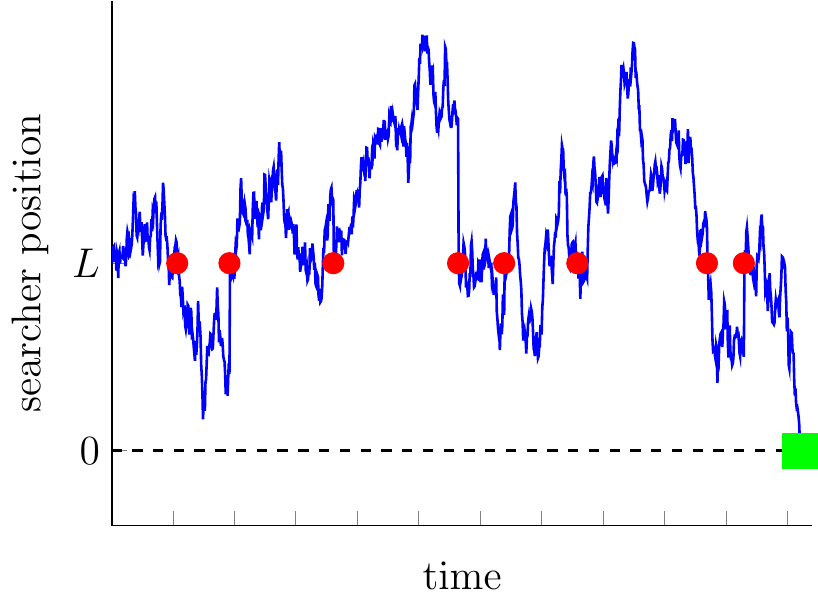}
             {\scriptsize(b)}\hspace{-.05\textwidth}
             \includegraphics[width=0.495\textwidth]{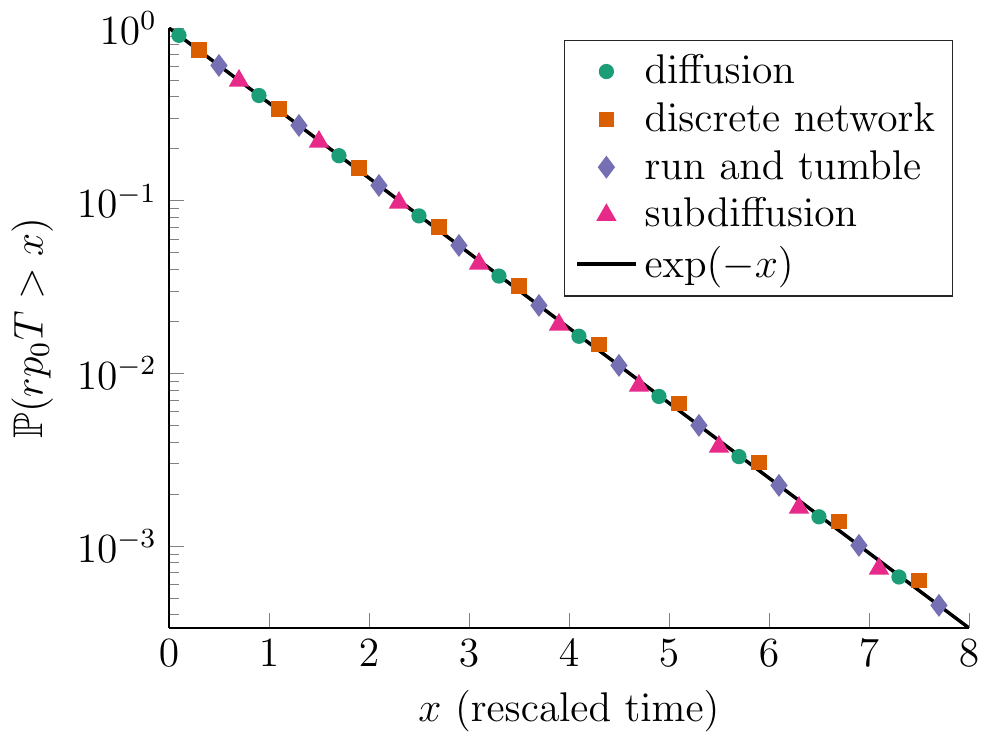}            
    \caption{{Panel (a)}: Diffusive search under exponential stochastic resetting with rate $r>0$ to $x=L>0$ in one spatial dimension. Red dots indicate resetting times and the green square indicates the FPT to $x=0$. {Panel (b)}: FPTs of disparate search processes behave similarly under frequent exponential stochastic resetting. The numerical details of the examples plotted here are described in section~\ref{examples}.}
 \label{figschem}
\end{figure}

Here, we study the full distribution and moments of FPTs for a wide class of stochastic search processes in the limit of frequent stochastic resetting. To briefly summarize our results, let $T$ denote the FPT of a stochastic searcher that resets to its initial position (or distribution of initial positions) at random times. We are most interested in the case that the resetting occurs at exponentially distributed times with rate $r>0$ (i.e.\ Poissonian resetting), but we prove our results for much more general resetting time distributions. If $p=p(r)\in(0,1)$ denotes the probability that the searcher finds a target before resetting, then under very mild assumptions on the search process we prove that $rp_{0}T$ converges in distribution to an exponential random variable with unit rate, which we denote by
\begin{align} \label{eq:1}
    rp_{0} T \to_{\textup{dist}} \textup{Exponential}(1) \quad\text{as }r\to\infty,
\end{align}
where $p_{0}=p_{0}(r)$ is any function of $r$ satisfying $p_{0}\sim p$ as $r\to\infty$. Throughout this work, $f\sim g$ denotes $f/g\to 1$. Roughly speaking, \eqref{eq:1} says that $T$ is approximately exponentially distributed with rate $rp_{0}$,
\begin{align*}
    \mathbb{P}(T>t) \approx e^{-rp_{0}t},\quad \text{if $t\ge 0$ and $p_{0}\ll1$}.
\end{align*}
In addition to the full distribution in \eqref{eq:1}, we determine the behavior of all the moments of $T$,
\begin{align} \label{eq:2}
    \mathbb{E}[T^m] \sim \frac{m!}{(rp_{0})^m} \quad\text{as }r\to\infty,
\end{align}
for integers $m\geq 1$. To make these results readily applicable, we determine appropriate choices of $p_{0}$ based on the short-time distribution of the search process without resetting.

The results in (\ref{eq:1}) and (\ref{eq:2}) show that many stochastic search processes with resetting behave similarly once we scale the search time by the resetting rate and the probability of a successful search. This is illustrated in Figure~\ref{figschem}b, which displays results from numerical solutions of quite disparate search processes that are nevertheless all approximately exponentially distributed with rate $rp_{0}$ (the details of these and other examples are given in section~\ref{examples}). Further, in typical scenarios of interest for diffusive search, we find that the asymptotic estimates converge exponentially fast. We note that we do not require the moments of the search time without resetting to be finite. Moreover, establishing these results requires knowledge only of the short-time behavior of the search process without resetting. With this information, we determine the asymptotic behavior of $p$ (i.e.\ we determine $p_{0}$) and thus the limiting distribution of the FPT.

The rest of the paper is organized as follows. In section \ref{results}, we prove \eqref{eq:1} and \eqref{eq:2}, and we determine the asymptotic behavior of $p$ under various assumptions on the short-time behavior of the search process without resetting. In section \ref{examples}, we apply these results to several scenarios including diffusive and subdiffusive search in one or three spatial dimensions, a random walk on a discrete network, and a run and tumble particle. We consider these examples in sections~\ref{diff}-\ref{sub} for the case of exponential resetting, and we consider diffusive search with sharp, uniform, and gamma resetting distributions in section~\ref{nonexp}. We conclude with a brief discussion. The Appendix contains proofs and technical points.

\section{Probabilistic setup and main results} \label{results}
In this section, we present results on the FPT of a searcher under frequent stochastic resetting. While these results make no explicit reference to the underlying search process (e.g.\ diffusion or otherwise), we later apply them to a diffusive search process and other stochastic processes.

\subsection{FPTs under frequent stochastic resetting}\label{mainsection}

Let $\tau$ denote a random variable whose cumulative distribution function, $F_{\tau}(t) \coloneqq  \mathbb{P}(\tau\leq t)$, satisfies
\begin{align} 
    F_{\tau}(0) &= 0,\label{Ftau0}\\ \quad F_{\tau}(t) &> 0 \quad\text{for some }t\in(0,\infty).\label{Ftau}
\end{align}
In words, \eqref{Ftau0} says that $\tau$ is strictly positive and \eqref{Ftau} merely excludes the trivial case that $\tau$ is always infinite. 

Let resetting occur according to a random variable, $\sigma>0$, with mean $\E[\sigma]=1/r$, where we refer to $r>0$ as the resetting rate. To construct $\sigma$ precisely, let $Y>0$ be any strictly positive random variable that does not depend on $r$ and has unit mean,
\begin{align}\label{Yunitmean}
\E[Y]=1,
\end{align}
and finite moment generating function in a neighborhood of the origin. That is, assume that there exists a $\delta>0$ so that
\begin{align}\label{ymgf}
\E[e^{zY}]<\infty\quad\text{for all }z\in[-\delta,\delta].
\end{align}
We then define $\sigma$ as
\begin{align}\label{sigmadef}
\sigma
:=Y/r.
\end{align}
The probability of the search process ending (i.e.\ a ``successful" search) prior to resetting is
\begin{align} \label{eq:p}
    p=p(r) \coloneqq  \mathbb{P}(\tau < \sigma) = \int_0^{\infty} S_{\sigma}(t) \,\text{d}F_{\tau}(t),
\end{align}
where $S_{\sigma}(t)=\P(\sigma>t)$ is the survival probability of $\sigma$. To exclude trivial cases, we assume
\begin{align}\label{nontrivial}
p>0\quad\text{for all }r>0.
\end{align}

Most prior studies of stochastic resetting consider exponential (i.e.\ Poissonian) resetting, which in this framework means that $Y$ is exponential with unit mean, and thus has survival probability
\begin{align}\label{expY}
S_{Y}(y)
:=\P(Y>y)
=e^{-y}\quad\text{for }y\ge0.
\end{align}
For such exponential resetting, $S_{\sigma}(t)=e^{-rt}$, $p$ in \eqref{eq:p} is the Laplace-Stieltjes transform of $F_{\tau}(t)$, and \eqref{Ftau} ensures \eqref{nontrivial} is satisfied.
However, the framework in \eqref{Yunitmean}-\eqref{sigmadef} includes much more general resetting distributions. For example, ``sharp reset'' \cite{pal2017} in which resetting occurs at a deterministic time $\sigma=1/r$ fits into this framework by setting
\begin{align*}
S_{Y}(y)
=\begin{cases}
1 & \text{if }y<1,\\
0 & \text{if }y\ge1.
\end{cases}
\end{align*}
Many other choices of the resetting time (such as uniform reset and gamma distributed reset considered in \cite{pal2017}) also fit into this framework (see section~\ref{nonexp} for examples).

Let $R\in\{0,1,\dots\}$ denote the number of resets before the searcher finds the target, or the number of ``unsuccessful" searches. From \eqref{eq:p}, we infer that $R$ is a geometric random variable with probability of success $p\in(0,1)$. That is,
\begin{align} \label{defR}
    \mathbb{P}(R=n) = (1-p)^np, \quad\text{for } n\in\{0,1,\dots\}.
\end{align}
The random variable that describes the total search time with resetting, $T>0$, is thus given by
\begin{align} \label{eq:T}
    T \coloneqq  \sum_{n=1}^R \sigma_n^- + \tau^-,
\end{align}
where $\{\sigma^-_n\}_{n\geq 1}$ is an independent and identically distributed sequence of random variables with common survival probability,
\begin{align} \label{eq:Ssigm}
    S_{\sigma^-}(t) \coloneqq \mathbb{P}(\sigma^- > t)  = \mathbb{P}(\sigma > t | \sigma < \tau).
\end{align}
Further, $\tau^-$ is a random variable defined by survival probability,
\begin{align*}
    S_{\tau^-}(t) \coloneqq \mathbb{P}(\tau^- > t)  = \mathbb{P}(\tau > t | \tau < \sigma).
\end{align*}
In words, the definition of $T$ in \eqref{eq:T} is the sum of the unsuccessful search times plus the single successful search time. We emphasize that $R$, $\{\sigma_{n}^-\}_{n\ge1}$, and $\tau^-$ are independent.

Before presenting the main results, we sketch the essential idea. In the frequent resetting limit ($r\to\infty$), the iteration of the search process that finally reaches the target before resetting, denoted by $\tau^-$ in (\ref{eq:T}), must be fast. Indeed, we show that this term is negligible compared to the other terms in \eqref{eq:T} for large $r$. Further, the condition imposed on the survival probability in \eqref{eq:Ssigm} becomes inconsequential for frequent resetting, and thus $\sigma^{-}\approx\sigma$. Therefore, \eqref{eq:T} reduces to a geometric sum of independent random variables, which becomes exponentially distributed by the so-called ``weak law of small numbers'' \cite{durrett2019}.

\begin{theorem} \label{main}
Under the assumptions of section~\ref{mainsection}, let $p_{0}=p_{0}(r)$ be any function of $r$ satisfying
\begin{align*}
p_{0}\sim p\quad\text{as }r\to\infty.
\end{align*}
Then $rp_{0}T$ converges in distribution to an exponential random variable with unit rate, 
\begin{align} \label{main1}
    rp_{0} T \to_{\textup{dist}} \textup{Exponential}(1)\quad \text{as }r\to\infty.
\end{align}
Further, for integers $m\geq 1$,
\begin{align} \label{main2}
    \mathbb{E}[T^m] \sim \frac{m!}{(rp_{0})^m} \quad\text{as }r\to\infty.
\end{align}
\end{theorem}

Recall that \eqref{main1} means 
\begin{align*}
    \lim_{r\to\infty} \mathbb{P}(rp_{0}T> x)
    = e^{-x},\quad\text{for all }x\ge0.
\end{align*}
We prove \eqref{main1} by proving convergence of moment generating functions and invoking L\'evy's continuity theorem. Now, convergence in distribution does not imply moment convergence. That is, \eqref{main1} does not imply \eqref{main2}. We prove \eqref{main2} by showing that $\{prT\}_{r>0}$ is uniformly integrable for sufficiently large $r$. The proof of Theorem~\ref{main} and the other results in this section are given in the Appendix.

\begin{remark}\label{remarkexact}
For the case of the first moment (i.e.\ $m=1$) for exponential resetting, one can verify that
\begin{align}\label{exact}
\E[T]
=\frac{1}{rp}(1-p)\quad\text{for any }r>0,
\end{align}
and thus the asymptotic result in \eqref{main2} follows from the exact representation in \eqref{exact} and the intuitive result that $\lim_{r\to\infty}p=0$ (see Lemma~\ref{pvanish} in the Appendix).
\end{remark}

\subsection{Asymptotics of the probability $p$ of a successful search for exponential resetting}\label{charsection}

Applying Theorem~\ref{main} to a given system requires knowledge of $p$. In this section, we assume exponential resetting (i.e.\ $S_{\sigma}(t)=e^{-rt}$) and consider various assumptions on the cumulative distribution function $F_{\tau}$ and determine the resulting asymptotics of the probability $p$ of a successful search, which then yields the distribution and moments of the FPT $T$ via Theorem~\ref{main}.

\subsubsection{Diffusion}

We first consider the typical short-time behavior of $F_{\tau}$ for a diffusive search when the searcher cannot start arbitrarily close to the target. In this case, $F_{\tau}(t)$ decays exponentially as $t\to0^{+}$. We begin with a result when we merely know the short-time behavior of $F_{\tau}$ on a logarithmic scale.

\begin{theorem} \label{thmlognormal}
For exponential resetting as in \eqref{expY}, if
\begin{align}\label{lim_thmlognormal}
    \lim_{t\to 0^+} t \ln F_{\tau}(t) = -C < 0,
\end{align}
then ${{p}}$ in \eqref{eq:p} and $T$ in \eqref{eq:T} satisfy
\begin{align*} 
\ln p
&\sim -\sqrt{4Cr}\quad\text{as }r\to\infty,\\
\ln (r^{m}\E[T^{m}])
&\sim m\sqrt{4Cr}\quad\text{as }r\to\infty.
\end{align*}
\end{theorem}

In applications of interest, the constant $C$ in \eqref{lim_thmlognormal} is a  characteristic timescale of diffusive search. Typically, $C$ is given by
\begin{align}\label{Ctypical}
    C = \frac{L^2}{4D}>0,
\end{align}
where $L>0$ is the shortest distance (in an appropriate distance metric) the searcher must travel to reach the target, and $D>0$ is a characteristic diffusion coefficient \cite{lawley2020uni}.

The next result yields stronger conclusions about the asymptotics of $p$ and $T$ by assuming more detailed information about the short-time behavior of $F_{\tau}$.

\begin{theorem} \label{thmlinearnormal}
For exponential resetting as in \eqref{expY}, if
\begin{align} \label{thmlinearnormalF}
    F_{\tau}(t) \sim At^be^{-C/t} \quad\text{as }t\to 0^+,
\end{align}
for $A>0$, $C>0$, and $b\in\mathbb{R}$, then ${{p}}$ in \eqref{eq:p} and $T$ in \eqref{eq:T} satisfy
\begin{align}
    p
    &\sim p_{0}:=A\sqrt{\pi C^{\frac{2b+1}{2}}} r^{\frac{1-2b}{4}} \exp(-\sqrt{4Cr}) \quad\text{as }r\to\infty,\label{p0diff}\\
    \mathbb{E}[T^{m}] 
    &\sim m!\bigg(\frac{1}{A\sqrt{\pi C^{\frac{2b+1}{2}}}} r^{\frac{2b-5}{4}}\exp(\sqrt{4Cr})\bigg)^{m} \quad\text{as }r\to\infty.\nonumber
\end{align}
\end{theorem}
In addition to the timescale $C$ in \eqref{Ctypical}, the constants $A$ and $b$ in \eqref{thmlinearnormalF} encode more details about the diffusive search process \cite{Linn_2022}. Examples of these constants in specific scenarios are given in section~\ref{diff}.

\subsubsection{Subdiffusion}

We now consider the case of subdiffusive search, where subdiffusion is modeled by a fractional Fokker-Planck equation \cite{metzler1999}. The results are analogous to the case of diffusion above, except the formulas are more complicated because the short-time behavior of $F_{\tau}$ depends on the subdiffusive exponent. We apply the following two theorems to specific examples in section~\ref{sub}.

\begin{theorem} \label{thmlog}
For exponential resetting as in \eqref{expY}, if
\begin{align} \label{lim_thmlog}
    \lim_{t\to 0^+} t^d \ln F_{\tau}(t) = -C < 0,
\end{align}
for $d>0$, then ${{p}}$ in \eqref{eq:p} and $T$ in \eqref{eq:T} satisfy
\begin{align}
\ln p
&\sim -\gamma r^{d/(d+1)}\quad\text{as }r\to\infty, \label{thmlogpdp}\\
\ln (r^{m}\E[T^{m}])
&\sim m\gamma r^{d/(d+1)}\quad\text{as }r\to\infty, \label{thmlogpdt}
\end{align}
where $\gamma = \frac{d+1}{d^{d/(d+1)}}C^{\frac{1}{d+1}}>0$.
\end{theorem}

For a subdiffusive exponent $\alpha\in(0,1)$, meaning that the mean-squared displacement of the search process without resetting grows sublinearly in time $t$ according to the power law $t^{\alpha}$, we typically have that \cite{lawley2020sub}
\begin{align}\label{alphastuff}
d=\frac{\alpha}{2-\alpha},\quad
C=(2-\alpha)\alpha^{\alpha/(2-\alpha)}\Big(\frac{L^{2}}{4K_{\alpha}}\Big)^{1/(2-\alpha)},
\end{align}
where the lengthscale $L>0$ is as in \eqref{Ctypical} and $K_{\alpha}$ is the characteristic subdiffusion exponent (with dimensions $(\text{length})^{2}(\text{time})^{-\alpha}$). Hence, the quantity $\gamma r^{d/(d+1)}$ appearing in \eqref{thmlogpdp}-\eqref{thmlogpdt} is typically given by
\begin{align*}
\gamma r^{d/(d+1)}
=\sqrt{r^{\alpha}L^{2}/K_{\alpha}}.
\end{align*}

\begin{theorem} \label{thm4}
For exponential resetting as in \eqref{expY}, if
\begin{align} \label{thm4F}
    F_{\tau}(t) \sim At^be^{-C/t^d} \quad\text{as }t\to 0^+,
\end{align}
for $A>0$, $C>0$, $d>0$, and $b\in\mathbb{R}$, then ${{p}}$ in \eqref{eq:p} and $T$ in \eqref{eq:T} satisfy
\begin{align}
    p
    &\sim p_{0}:={\mu} r^{\beta} \exp(-\gamma r^{d/(d+1)}) \quad\text{as }r\to\infty,\\
    \mathbb{E}[T^{m}]
    &\sim m!\Big(\frac{1}{{\mu}} r^{-\beta-1}\exp(\gamma r^{d/(d+1)})\Big)^{m}\quad\text{as }r\to\infty,
\end{align}
where
\begin{align*}
    {\mu} = \sqrt{\frac{2\pi A^2 (Cd)^{\frac{2b+1}{d+1}}}{d+1}}, \quad \beta = \frac{d-2b}{2d+2}, \quad \gamma = \frac{d+1}{d^{d/(d+1)}}C^{\frac{1}{d+1}}.
\end{align*}
\end{theorem}
\noindent
If $d=1$ in \eqref{lim_thmlog} and \eqref{thm4F} (and thus $\alpha=1$ in \eqref{alphastuff}), Theorem~\ref{thmlog} reduces to Theorem~\ref{thmlognormal} and Theorem~\ref{thm4} reduces to Theorem~\ref{thmlinearnormal}. 

\subsubsection{Other search processes}

We now consider the case that $F_{\tau}(t)$ decays according to a power law as $t\to0^{+}$. This can describe the case of (i) diffusive search in which the searcher can start arbitrarily close to the target (see section~\ref{unif}), (ii) search by a continuous-time Markov chain on a discrete state space (see section~\ref{ctmc}), and (iii) superdiffusive search (see \cite{lawley2021super}).

\begin{proposition}\label{pfeller}
For exponential resetting as in \eqref{expY}, if
\begin{align} \label{FFeller}
    F_{\tau}(t) \sim At^b \quad\text{as }t\to 0^+,
\end{align}
for $A>0$ and $b>0$, then $p$ in \eqref{eq:p} and $T$ in \eqref{eq:T} satisfy
\begin{align}
    p
    &\sim p_{0}:=\Gamma(b+1)Ar^{-b}\quad\text{as }r\to\infty,\label{fellerp}\\
    \mathbb{E}[T^{m}]
    &\sim m!\Big(\frac{1}{A\Gamma(b+1)} r^{b-1}\Big)^{m}\quad\text{as }r\to\infty,\label{fellerT}
\end{align}
where $\Gamma(\beta) \coloneqq  \int_0^{\infty} z^{\beta - 1}e^{-z}\,\text{d}z$ denotes the gamma function.
\end{proposition}

The asymptotics of $p$ in \eqref{fellerp} follow from noticing that $p$ is the Laplace-Stieltjes transform of $F_{\tau}(t)$ and applying Tauberian theorems (see, for example, Theorem~3 in section 5 of chapter 8 of \cite{feller-vol-2}). The asymptotics of the moments of $T$ in \eqref{fellerT} then follow from Theorem~\ref{main}.

\section{Examples \& numerical solutions} \label{examples}
The results in section~\ref{results} give the asymptotics of the FPT $T$ as the resetting rate $r$ increases. In this section, we apply these results to several examples and compare them to numerical solutions. The details of the calculations for these examples are given in the Appendix. We assume exponential resetting as in \eqref{expY} in sections~\ref{diff}-\ref{sub} and consider non-exponential resetting in section~\ref{nonexp}.

\subsection{Diffusive search}\label{diff}

Consider a searcher that diffuses with diffusivity $D>0$ in $d\ge1$ spatial dimensions. Assume that the searcher starts at (and is reset to) a position that is distance $L>0$ from the target. Consider the following three scenarios: (i) $d=1$ and the target is a single point, (ii) $d=3$ and the target is a sphere of radius $a>0$, and (iii) $d=3$ and the target is the exterior of a sphere centered at the starting and resetting position (i.e.\ the FPT is the first time the searcher escapes a sphere of radius $L>0$).

In each of these three scenarios, the Laplace transform of the distribution of $T$ and all the moments of $T$ can be calculated analytically. Further, the probability $p$ in \eqref{eq:p} of a successful search and the corresponding asymptotic form $p_{0}$ given by Theorem~\ref{thmlinearnormal} can be computed analytically. In particular, the respective values of $p_{0}$ in \eqref{p0diff} in Theorem~\ref{thmlinearnormal} for these three scenarios are
\begin{align}\label{p0sdiff}
p_{0}=e^{-\sqrt{rL^{2}/D}}, \quad
p_{0}=\frac{1}{1+L/a}e^{-\sqrt{rL^{2}/D}}, \quad
p_{0}=\sqrt{\frac{4rL^{2}}{D}}e^{-\sqrt{rL^{2}/D}}.
\end{align}
These calculations are given in the Appendix.

\begin{figure}
\centering
{\scriptsize(a)}\hspace{-.05\textwidth}
\includegraphics[width=0.495\textwidth]{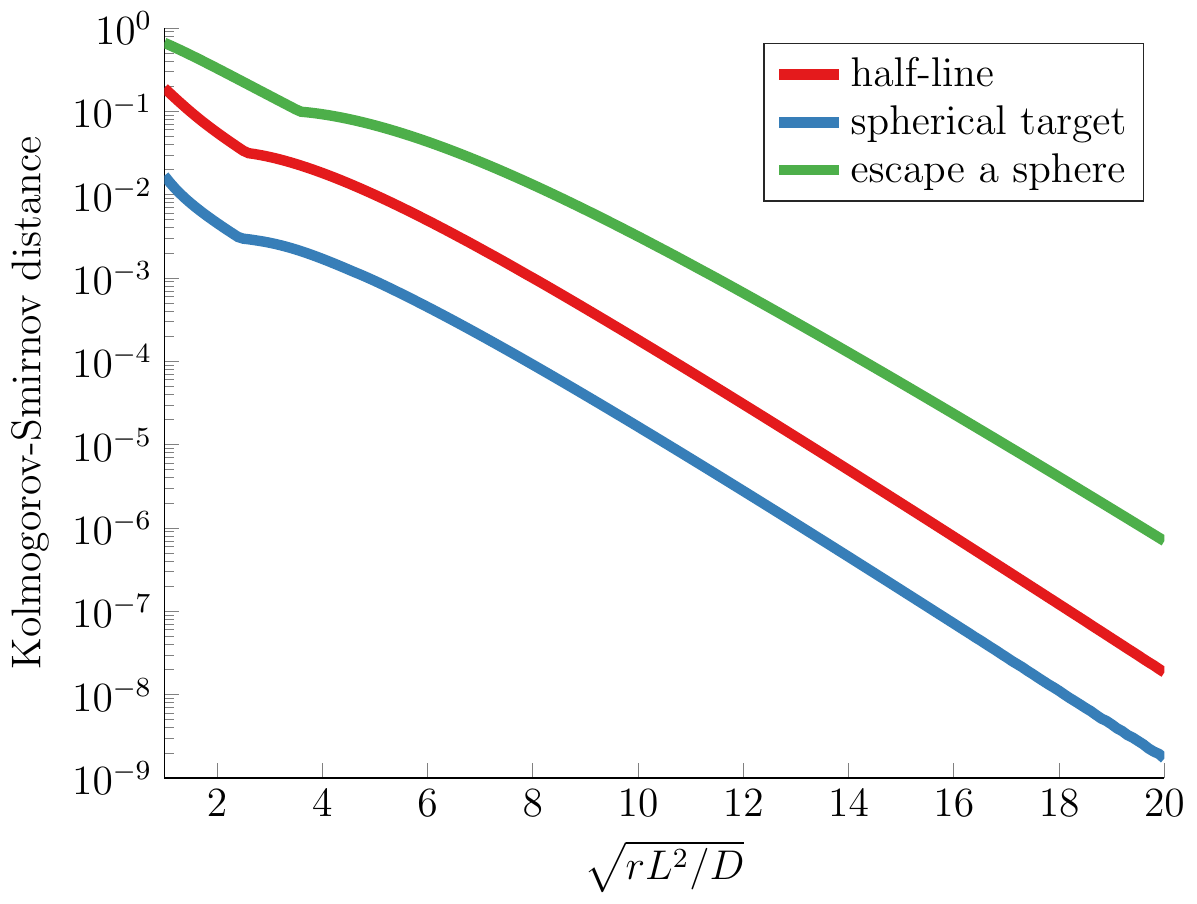}
{\scriptsize(b)}\hspace{-.05\textwidth}
\includegraphics[width=0.495\textwidth]{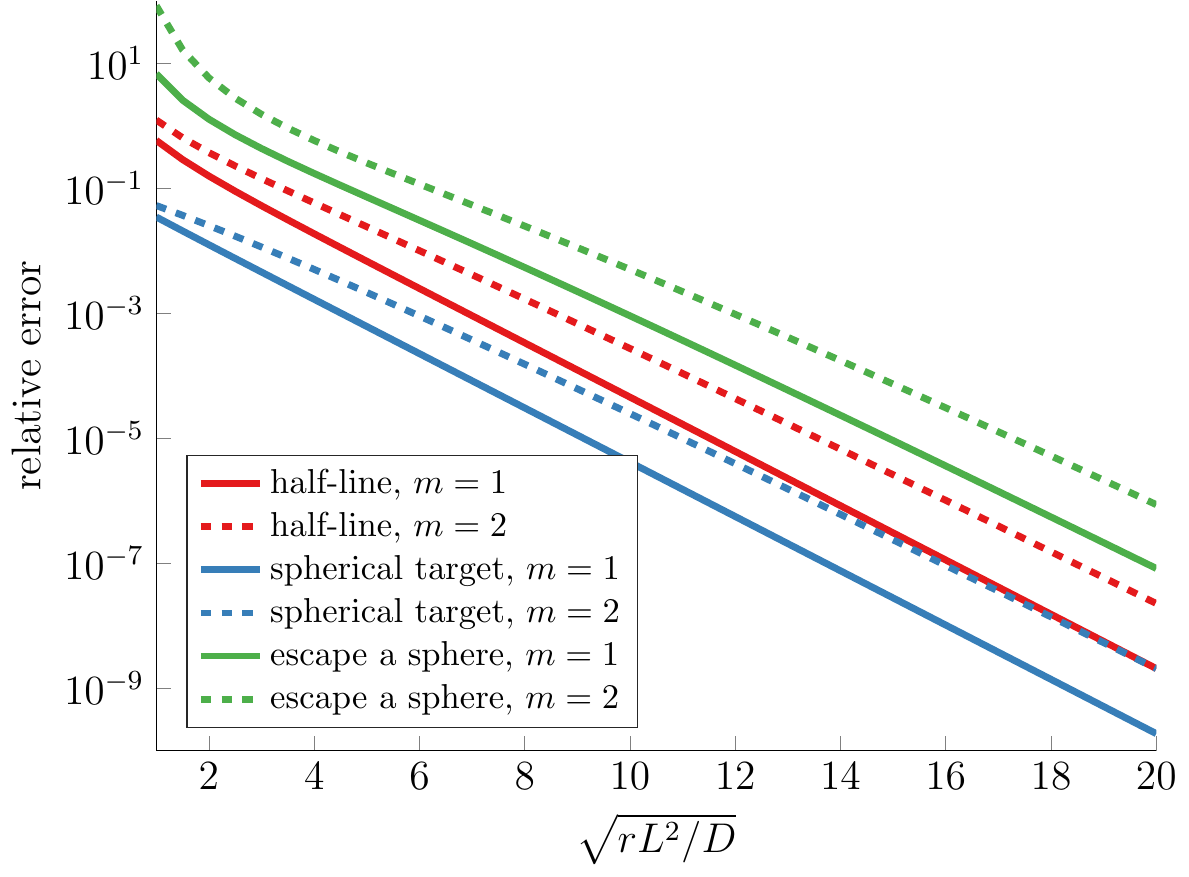}
 \caption{Diffusive search for a target that is distance $L>0$ from the initial and resetting position. See section~\ref{diff} for details.}
 \label{figdiff}
\end{figure}

In Figure~\ref{figdiff}a, we plot the convergence in distribution of $rp_{0}T$ to a unit rate exponential random variable as the dimensionless resetting rate $\sqrt{rL^{2}/D}$ increases for each of these examples. Specifically, we plot the Kolmogorov-Smirnov distance between the distribution of $rp_{0}T$ and a unit rate exponential distribution, defined as
\begin{align}\label{ks}
\sup_{x\ge0}\big|\P(rp_{0}T>x)-e^{-x}\big|.
\end{align}
In agreement with Theorem~\ref{main}, $rp_{0}T$ rapidly converges in distribution to a unit rate exponential random variable. In Figure~\ref{figdiff}b, we plot the relative error between the exact $m$th moment $\E[T^{m}]$ and the frequent resetting estimate $\E[T^{m}]\approx m!/(rp_{0})^{m}$ from Theorems~\ref{main} and \ref{thmlinearnormal} for $m=1$ (solid curves) and $m=2$ (dashed curves). This figure shows that the relative error vanishes exponentially fast as $\sqrt{rL^{2}/D}$ increases. 

These three examples share the general features that if the resetting rate $r$ is much faster than the diffusion rate  (i.e. $r\gg D/L^{2}$), then $T$ is approximately exponentially distributed with rate $rp_{0}$, where $p_{0}$ vanishes exponentially according to $p_{0}\approx e^{-\sqrt{rL^{2}/D}}$ (possibly with a pre-factor that depends on the details of the geometry). While these features can be seen explicitly in the three analytically tractable scenarios described above, they characterize diffusive search with frequent resetting much more generally. Indeed, as long as the searcher cannot start (or reset) arbitrarily close to the target (see section~\ref{unif} for a case that this condition excludes), then the FPT distribution without resetting generally satisfies \cite{lawley2020uni}
\begin{align*}
\lim_{t\to0^{+}}t\ln F_{\tau}(t)
=-L^{2}/D<0,
\end{align*}
where $D>0$ is a characteristic diffusivity and $L>0$ is the shortest distance from the set of initial positions to the target (in an appropriate distance metric). Hence, Theorem~\ref{thmlognormal} yields
\begin{align*}
\ln p\sim-\sqrt{rL^{2}/D}<0\quad\text{as }r\to\infty,
\end{align*}
and thus the moments of $T$ diverge exponentially according to
\begin{align*}
\ln(r^{m}\E[T^{m}])
\sim m\sqrt{rL^{2}/D}\quad\text{as }r\to\infty.
\end{align*}

\subsection{Diffusive search with uniform initial condition}\label{unif}

We now consider diffusive search in which the searcher's starting and resetting positions are not bounded away from the target. Suppose that the searcher diffuses with diffusivity $D>0$ in one spatial dimension with targets at $x=0$ and $x=L>0$, and suppose that the searcher starts and resets to a uniformly distributed position in the interval $[0,L]$. In this case, the FPT distribution without resetting decays according to the following power law at short time,
\begin{align}\label{shortunif}
F_{\tau}(t)
\sim\sqrt{\frac{16 Dt}{\pi L^{2}}}\quad\text{as }t\to0^{+}.
\end{align}
Hence, for frequent resetting, Theorem~\ref{main} implies that $T$ is approximately exponentially distributed with rate $rp_{0}$, where Proposition~\ref{pfeller} yields
\begin{align}\label{p0unifexample}
p\sim p_{0}
=\sqrt{4D/(rL^{2})}\quad\text{as }r\to\infty.
\end{align}
In Figure~\ref{figunif}a, we plot the Kolmogorov-Smirnov distance as in \eqref{ks} for this example (red curve) as the dimensionless resetting rate $\sqrt{rL^{2}/D}$ increases. In Figure~\ref{figunif}a,  we also plot the Kolmogorov-Smirnov distance for scenario (iii) in section~\ref{diff} above (blue curve) except where the searcher starts and resets to uniformly distributed positions in the sphere.

\begin{figure}
  \centering
             {\scriptsize(a)}\hspace{-.05\textwidth}
             \includegraphics[width=0.495\textwidth]{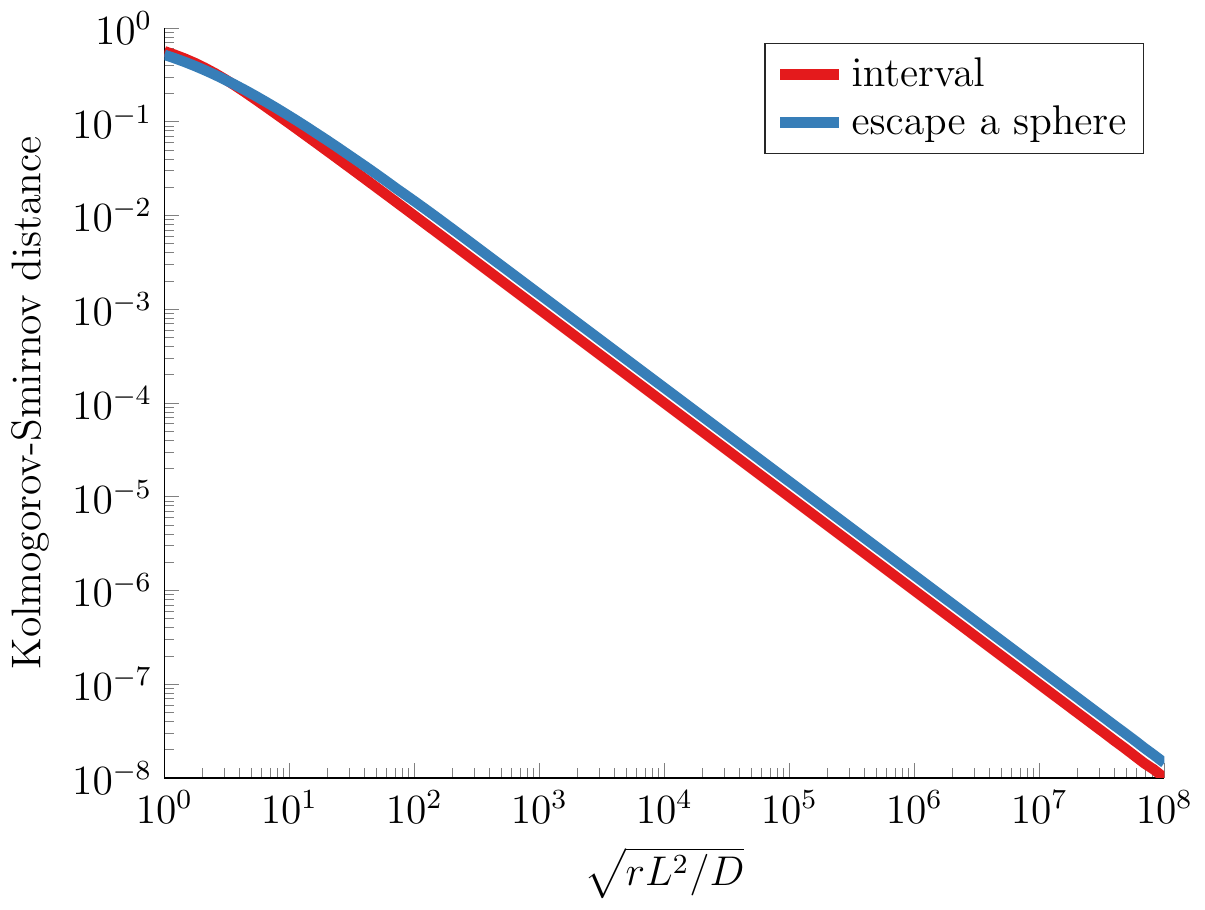}
             {\scriptsize(b)}\hspace{-.05\textwidth}
             \includegraphics[width=0.495\textwidth]{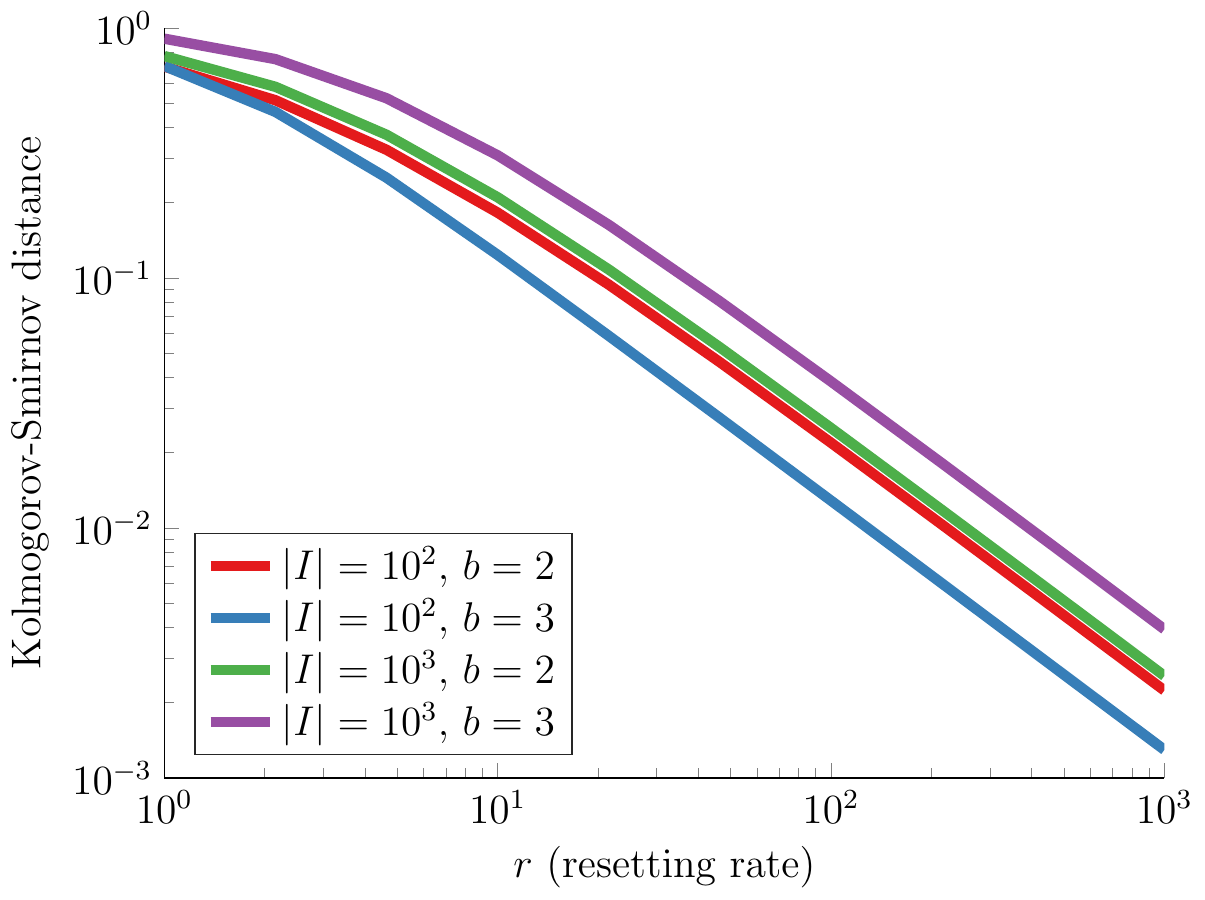}
 \caption{{Panel (a)}: Diffusive search with uniformly distributed initial and resetting positions. {Panel (b)}: Search on a discrete network with $|I|\gg1$ states and $b\ge1$ jumps required to reach the target from the starting and resetting position. See sections~\ref{unif}-\ref{ctmc} for details.}
 \label{figunif}
\end{figure}

\subsection{Search on a discrete network}\label{ctmc}

Suppose the searcher moves by discrete jumps to adjacent nodes on a discrete network \cite{bonomo2021}. Specifically, let $X=\{X(t)\}_{t\ge0}$ be a continuous-time Markov chain on a finite or countably infinite state space $I$. Suppose $X$ starts at (and resets at rate $r$ to) a given state $i_{0}\in I$. Consider the FPT to some target set of states $I_{\t}\subset I$ with $i_{0}\notin I_{\t}$.

Proposition~1 in \cite{lawley2020networks} implies that the cumulative distribution function of the FPT $\tau:=\inf\{t>0:X(t)\in I_{\t}\}$ without resetting decays according to the following power law at short-time,
\begin{align}\label{Fctmc}
F_{\tau}(t)
\sim At^{b}\quad\text{as }t\to0^{+},
\end{align}
where $b\ge1$ is the minimum number of jumps $X$ must take to reach $I_{\t}$ from $i_{0}$ and $A=\Lambda/b!$, where $\Lambda$ is the product of the jump rates of $X$ along this shortest path from $i_{0}$ to $I_{\t}$. (If there are multiple shortest paths, then $\Lambda$ is the sum of the products of the jump rates along these paths.) As a technical aside, \eqref{Fctmc} assumes that the jump rates of $X$ are bounded and $\P(\tau=\infty)\neq1$ (i.e.\ there exists a path from $i_{0}$ to $I_{\t}$).

Hence, Proposition~\ref{pfeller} implies that 
\begin{align*}
p
\sim p_{0}=A\Gamma(b+1) r^{-b}\quad\text{as }r\to\infty,
\end{align*}
and Theorem~\ref{main} yields the distribution and moments of $T$ for frequent resetting. In Figure~\ref{figunif}, we plot the Kolmogorov-Smirnov distance as in \eqref{ks} for this example as the resetting rate $r$ increases for a few different randomly generated networks with the number of states ranging from $|I|=10^{2}$ to $|I|=10^{3}$. The details of this calculation and these networks are given in the Appendix. The convergence in distribution illustrated in Figure~\ref{figunif} shows that despite the complexity of these underlying jump processes, the FPT with frequent resetting depends only on the network properties along the shortest path(s) to the target.

\subsection{Run and tumble}\label{rtp}

Consider a one-dimensional run and tumble particle that switches between velocity $V>0$ and $-V<0$ at Poissonian rate $\lambda>0$ \cite{evans2018, santra2020}. Analogous to the first scenario in section~\ref{diff} on diffusion, suppose the target is at $x=0$ and the searcher starts at and resets to $x=L>0$. The probability of a successful search satisfies
\begin{align}\label{p0rtp}
p
\sim p_{0}
=e^{-r L/V}(\tfrac{1}{2}e^{-\lambda L/V}+\beta/r)\quad\text{as }r\to\infty,
\end{align}
where $\beta=\lambda  e^{-\frac{\lambda  L}{V}} (\lambda  L+V)/(4V)$.
In Figure~\ref{figrtp}a, we plot the Kolmogorov-Smirnov distance as in \eqref{ks} for this example as the dimensionless resetting rate $rL/V$ increases for a few different choices of the dimensionless tumbling rate $\lambda L/V$. The details of this calculation are given in the Appendix.

\begin{figure}
  \centering
             {\scriptsize(a)}\hspace{-.05\textwidth}
             \includegraphics[width=0.495\textwidth]{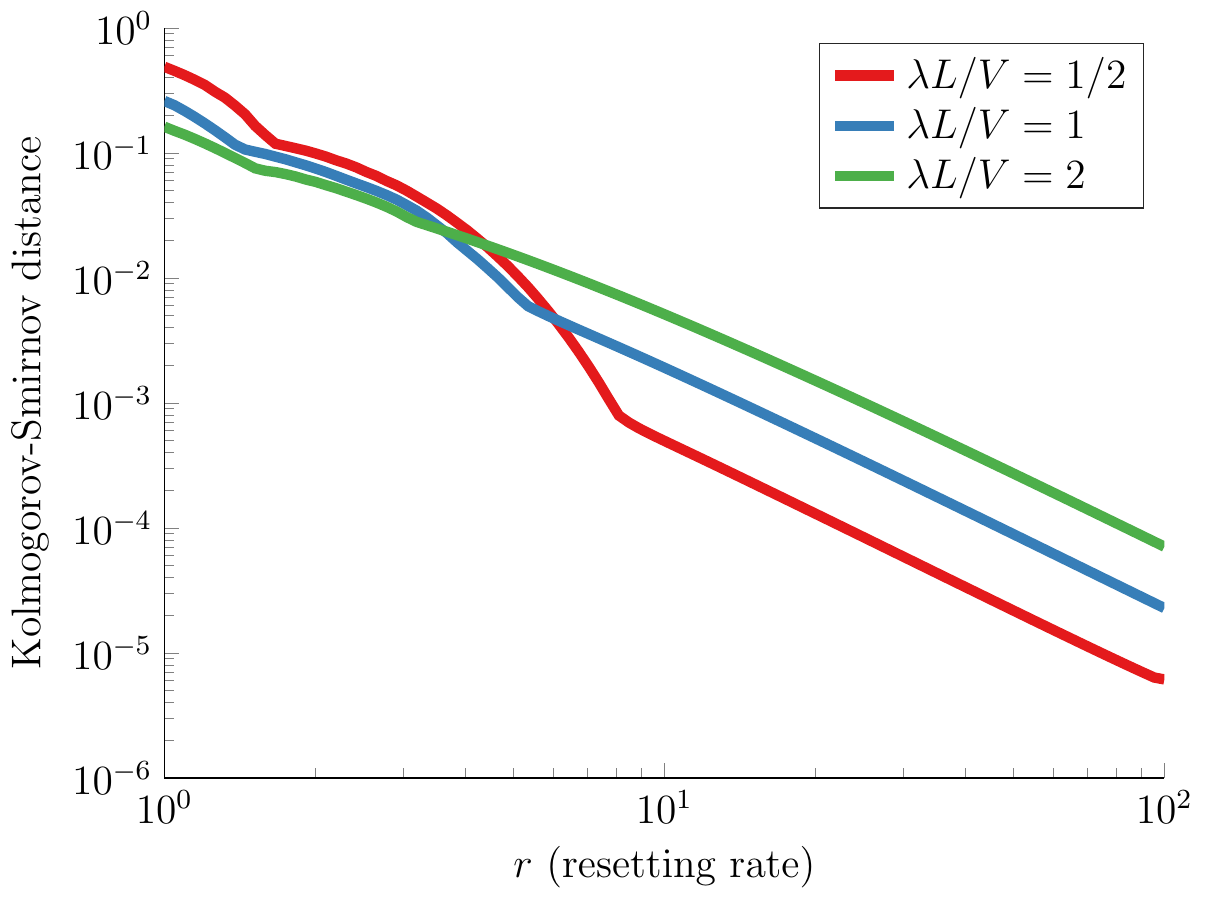}
             {\scriptsize(b)}\hspace{-.05\textwidth}
             \includegraphics[width=0.495\textwidth]{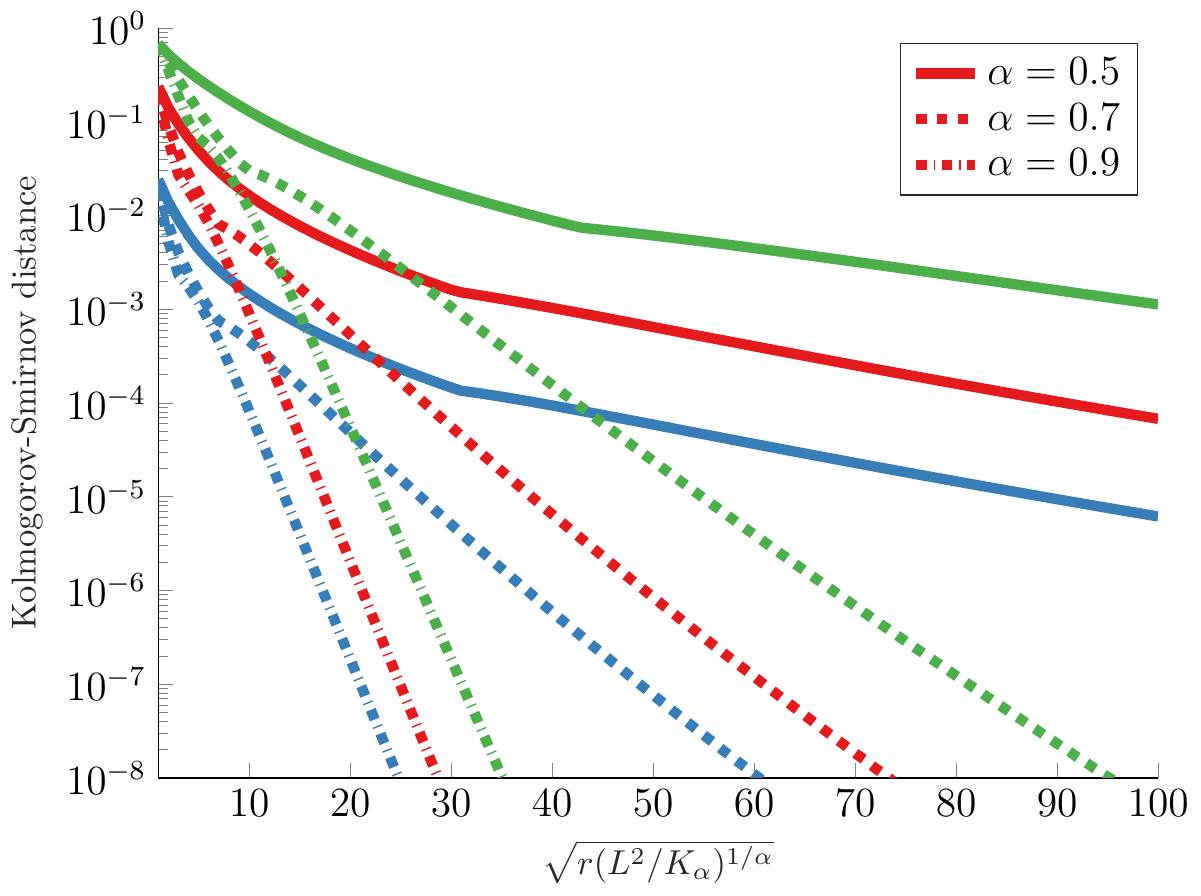}
 \caption{{Panel (a)}: Run and tumble search. {Panel (b)}: Subdiffusive search. See sections~\ref{rtp}-\ref{sub} for details.}
 \label{figrtp}
\end{figure}

\subsection{Subdiffusive search}\label{sub}

Suppose that the searcher moves by subdiffusion in $d\ge1$ spatial dimensions between stochastic resets, meaning that its mean-squared displacement grows in time $t$ according to a sublinear power law $t^{\alpha}$ for $\alpha\in(0,1)$. Concretely, suppose that in between stochastic resets, the probability density $p_{\alpha}(x,t)$ for its position evolves according to the following fractional Fokker-Planck equation \cite{metzler1999},
\begin{align}\label{fde}
\frac{\partial}{\partial t}p_{\alpha}(x,t)
=\D K_{\alpha}\Delta p_{\alpha}(x,t),\quad t>0,\,x\in\R^{d}\backslash U,
\end{align}
with absorbing conditions on the target $U\subset \R^{d}$,
\begin{align*}
p_{\alpha}=0,\quad x\in U.
\end{align*}
In \eqref{fde}, $t>0$ is the time elapsed since the last reset, $K_{\alpha}>0$ is the generalized diffusivity (with dimensions $(\text{length})^{2}(\text{time})^{-\alpha}$), and $\D$ is the Riemann-Liouville fractional derivative \cite{samko1993},
\begin{align*}
\D f(t)
=\frac{1}{\Gamma(\alpha)}\frac{\dd}{\dd t}\int_{0}^{t}\frac{f(s)}{(t-s)^{1-\alpha}}\,\dd s,
\end{align*}
where $\Gamma(\alpha)$  is the Gamma function. 

For a given initial distribution of the searcher and a given target, let $\tau_{\alpha}$ denote the FPT without stochastic resetting with corresponding survival probability,
\begin{align}\label{Sint}
S_{\alpha}(t)
:=\P(\tau_{\alpha}>t)
=\int_{\R^{d}\backslash U} p_{\alpha}(x,t)\,\dd x.
\end{align}
Now, the Laplace transform of the solution $p_{\alpha}$ to the fractional equation \eqref{fde} with $\alpha\in(0,1)$ is related to the solution $p_{1}$ to the integer order version of \eqref{fde} (i.e.\ \eqref{fde} with $\alpha=1$) via the relation \cite{lawley2020sr2}
\begin{align}\label{sr2rela}
\widetilde{p}_{\alpha}(x,s)
=s^{\alpha-1}\widetilde{p}_{1}(x,s^{\alpha}),\quad \alpha\in(0,1],
\end{align}
where $\widetilde{f}(s):=\int_{0}^{\infty}e^{-st}f(t)\,\dd t$ denotes Laplace transform. Hence, integrating \eqref{sr2rela} over space and using \eqref{Sint} yields
\begin{align}\label{relaalpha}
\widetilde{S}_{\alpha}(s)
=s^{\alpha-1}\widetilde{S}_{1}(s^{\alpha}),\quad \alpha\in(0,1].
\end{align}
In words, \eqref{relaalpha} yields the survival probability in Laplace space for subdiffusion, $\widetilde{S}_{\alpha}$, from the survival probability in Laplace space for normal diffusion, $\widetilde{S}_{1}$.

Using \eqref{relaalpha}, we now consider the three scenarios analyzed in section~\ref{diff}, but for subdiffusion. Specifically, we consider (i) $d=1$ and the target is a single point, (ii) $d=3$ and the target is a sphere of radius $a>0$, and (iii) $d=3$ and the target is the exterior of a sphere centered at the starting and resetting position. Using Theorem~\ref{thm4}, we obtain that the asymptotic forms $p_{0}$ of the probability of a successful search for these three scenarios are given respectively by
\begin{align*}
p_{0}
=e^{-\sqrt{r^{\alpha}L^{2}/K_{\alpha}}},\quad
p_{0}
=\frac{1}{1+L/a}e^{-\sqrt{r^{\alpha}L^{2}/K_{\alpha}}},\quad
p_{0}
=\sqrt{\frac{4r^{\alpha}L^{2}}{K_{\alpha}}}e^{-\sqrt{r^{\alpha}L^{2}/K_{\alpha}}}.
\end{align*}
In Figure~\ref{figrtp}b, we plot the Kolmogorov-Smirnov distance as in \eqref{ks} for these examples as the dimensionless resetting rate $\sqrt{r(L^{2}/K_{\alpha})^{1/\alpha}}$ increases for a few different choices of the subdiffusive exponent $\alpha\in(0,1)$. Scenarios (i), (ii), and (iii) correspond respectively to the red, blue, and green curves. The values $\alpha=0.5$, $\alpha=0.7$, and $\alpha=0.9$ correspond respectively to the solid, dashed, and dot dashed curves.

\subsection{Non-exponential resetting}\label{nonexp}

In the examples above, the resetting times were exponentially distributed with rate $r$. We now illustrate our results for other choices of the resetting time distribution. In particular, we define the resetting time via $\sigma=Y/r$ for the following three choices of the survival probability of $Y$,
\begin{align}
S_{Y}(y)
&=\begin{cases}
1 & \text{if }y<1,\\
0 & \text{if }y\ge1.
\end{cases}\quad\text{(sharp reset)},\label{nonexp1}\\
S_{Y}(y)
&=1-y/2,\quad y\in[0,2]\quad\text{(uniform reset)},\label{nonexp2}\\
S_{Y}(y)
&=(1+2y)e^{-2y},\quad y\ge0\quad\text{(gamma reset)}.\label{nonexp3}
\end{align}
The choice in \eqref{nonexp1} yields resetting times which are $\sigma=1/r$ with probability one (so-called sharp reset or sharp restart \cite{pal2017}). The choice in \eqref{nonexp2} yields resetting times $\sigma$ which are uniformly distributed on the interval $[0,2r]$. The choice in \eqref{nonexp3} yields resetting times $\sigma$ which have a gamma distribution with shape parameter $2$ and scale parameter $1/(2r)$.

We consider these three choices of the resetting time distribution for the case of diffusive search on the half-line (scenario (i) in section~\ref{diff} above). In Figure~\ref{fignonexp}, we plot the Kolmogorov-Smirnov distance as in \eqref{ks} (with $p_{0}=p$) as the dimensionless resetting rate $\sqrt{rL^{2}/D}$ increases for each of the three resetting time distributions in \eqref{nonexp1}-\eqref{nonexp3}. The distributions of $T$ for each marker in this plot are computed from $10^{7}$ realizations of $T$. The computation of $p$ for each of these examples is given in the Appendix.

\begin{figure}
  \centering
             \includegraphics[width=0.495\textwidth]{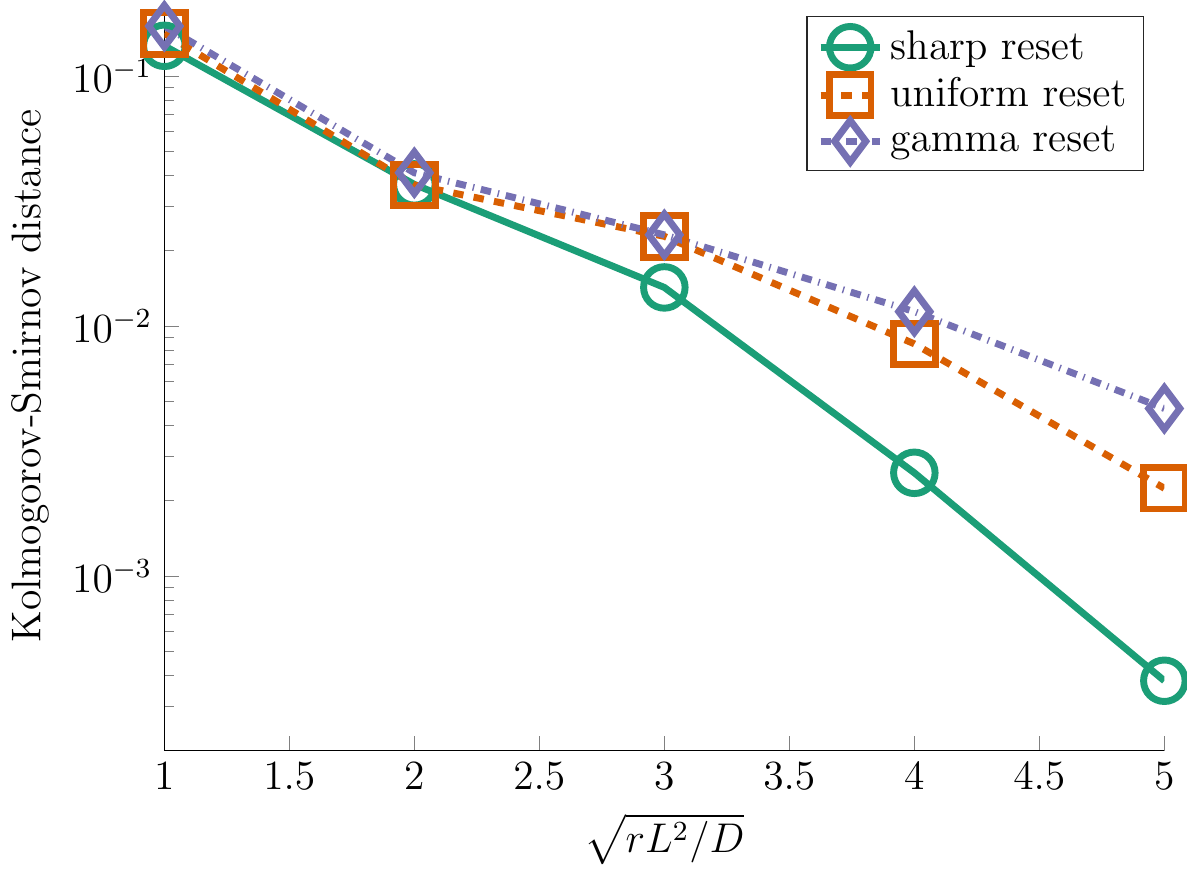}
 \caption{Non-exponentially distributed resetting times. See section~\ref{nonexp} for details.}
 \label{fignonexp}
\end{figure}

\section{Discussion}
In this work, we studied FPTs for stochastic search processes in the limit of frequent stochastic resetting. We determined approximations for the full probability distribution and moments of the FPT, which are exact in the frequent resetting limit. While we generally focused on the case that the resetting occurs at exponentially distributed times with rate $r>0$, we proved our results for much more general resetting time distributions. These results depend only on the short-time behavior of the search process without resetting and are thus immediately tractable in many settings. In particular, much of the details about the particular search process and geometry are irrelevant for frequent resetting, as illustrated in section~\ref{examples}. The relevant information about the search process is encoded into the quantity $p$, which is the probability of a successful search prior to resetting. We computed approximations $p_{0}\approx p$ for a variety of search processes for frequent exponential (Poissonian) resetting. By considering several specific examples with numerical solutions, we found that the error in these approximations often decays rapidly with the resetting rate $r$. These results show that some resetting search processes for which computing the exact distribution and statistics of the FPT is intractable may be well-approximated by simple asymptotic formulae.


\section{Appendix}

\subsection{Proofs of results in section~\ref{mainsection}}
We begin by proving the main results on distribution and moment convergence of the FPT with stochastic resetting. These proofs make use of subsidiary results, which we formalize as lemmas and whose proofs follow.
\begin{proof}[Proof of Theorem~\ref{main}]
By Lemma~\ref{lemmap0}, it is enough to prove the theorem for $p_{0}=p$. Fix $z\in(-\delta,\delta)$ with $\delta$ as in \eqref{ymgf}. Since $\sigma^-$ and $\tau^-$ are independent, the moment generating function of $prT$ can be expanded to
\begin{align} \label{eq:p1a}
    \mathbb{E}[e^{zpr T}] = \mathbb{E}\Bigg[\sum_{k=0}^{\infty} e^{zpr\sum_{n=1}^{k} \sigma_n^-} \mathbbm{1}_{R=k} \Bigg] \mathbb{E}[e^{zpr \tau^-}],
\end{align}
where $\mathbbm{1}_A$ denotes the indicator function on an event $A$, so $\mathbb{E}(\mathbbm{1}_A)= \mathbb{P}(A)$. Moreover, since $R$ is geometrically distributed with probability of success $p$ and $\{\sigma^-_n\}_{n\geq 1}$ are independent and identically distributed,
\begin{align} \label{eq:p1b}
    \mathbb{E}\Bigg[\sum_{k=0}^{\infty} e^{zpr\sum_{n=1}^k \sigma_n^-} \mathbbm{1}_{R=k} \Bigg] = \sum_{k=0}^{\infty} \big(\mathbb{E}[e^{zp
    r\sigma^-}]\big)^k(1-p)^kp.
\end{align}
Substituting (\ref{eq:p1b}) into (\ref{eq:p1a}), which is now in the form of a geometric series,
\begin{align} \label{eq:p1c}
    \mathbb{E}[e^{zpr T}] = \frac{p \mathbb{E}[e^{zpr \tau^-}]}{1 - (1-p)\mathbb{E}[e^{zpr\sigma^-}]}.
\end{align}
Using Lemma \ref{lem6} and Lemma \ref{lem7} when taking the limit as $r\to\infty$ of (\ref{eq:p1c}) yields
\begin{align} \label{eq:p1d}
    \lim_{r\to\infty} \mathbb{E}[e^{zpr T}] = \frac{1}{1-z}.
\end{align}
The right-hand side of (\ref{eq:p1d}) is the moment generating function of an exponential random variable with unit rate. By L\'evy's continuity theorem, convergence of moment generating functions implies convergence in distribution, and so the proof of \eqref{main1} in Theorem~\ref{main} is complete.

Now, since $\sigma^-$ and $\tau^-$ are independent, the second moment of $T$ can be written as
\begin{align} \label{eq:p2a}
    \mathbb{E}[T^2] = \mathbb{E}\Bigg[\Big(\sum_{n=1}^R \sigma_n^- \Big)^2\Bigg] + 2\mathbb{E}[\tau^-]\mathbb{E}\Bigg[\sum_{n=1}^R \sigma_n^-\Bigg] + \mathbb{E}[(\tau^-)^2].
\end{align}
Expanding the first term in \eqref{eq:p2a} yields
\begin{align} \label{eq:p2b}
    \mathbb{E}\Bigg[\Big(\sum_{n=1}^R \sigma_n^-\Big)^2\Bigg] &= \mathbb{E}\Bigg[ \sum_{n=1}^R (\sigma_n^-)^2\Bigg] + \mathbb{E}\Bigg[ \sum_{i=1}^R \sum_{\substack{j=1\\j\neq i}}^R \sigma^-_i \sigma^-_j\Bigg].
\end{align}
Since $R$ and $\sigma^-$ are independent, we apply Wald's Theorem to expand \eqref{eq:p2b},
\begin{align} \label{eq:p2b2}
    \mathbb{E}\Bigg[\Big(\sum_{n=1}^R \sigma_n^-\Big)^2\Bigg] = \mathbb{E}[R]\mathbb{E} [(\sigma^-)^2] + \mathbb{E}[R(R-1)] (\mathbb{E}[\sigma^-])^2. 
\end{align}
By Lemma~\ref{momentsequiv} and the result from \eqref{defR} that
\begin{align*}
\mathbb{E}[R] = \frac{1-p}{p},\quad \mathbb{E}[R^2] = \frac{(2-p)(1-p)}{p^2},
\end{align*}
we find
\begin{align*}
    \lim_{r\to\infty} (pr)^2\mathbb{E}\Bigg[\Big(\sum_{n=1}^R \sigma_n^-\Big)^2\Bigg] = 2.
\end{align*}
To address the other terms in (\ref{eq:p2a}), we apply the Cauchy-Schwarz inequality,
\begin{align} \label{eq:p2d}
    \mathbb{E}[(\tau^-)^m] = \frac{\mathbb{E}[\tau^m \mathbbm{1}_{\tau<\sigma}]}{p} \leq \frac{\mathbb{E}[\sigma^m \mathbbm{1}_{\tau<\sigma}]}{p} \leq \frac{\sqrt{\mathbb{E}[\sigma^{2m}]} \sqrt{\mathbb{E}[(\mathbbm{1}_{\tau<\sigma})^2]}}{p} = \frac{\sqrt{\E[Y^{2m}]}}{r^m\sqrt{p}}.
\end{align}
From this, we conclude
\begin{align} \label{eq:p2e}
    \lim_{r\to\infty} (pr)^2 \mathbb{E}[\tau^-] \mathbb{E}\Bigg[\sum_{n=1}^R \sigma_n^-\Bigg] = \lim_{r\to\infty} (pr)^2 \mathbb{E}[(\tau^-)^2] = 0.
\end{align}
Combining (\ref{eq:p2d}) with (\ref{eq:p2e}) and that $p\to 0$ as $r\to\infty$ by Lemma \ref{pvanish}, we conclude $\lim_{r\to\infty} \mathbb{E}[(pr T)^2] =2$. This implies $\{prT\}_{r>0}$ is uniformly integrable for sufficiently large $r$ (see, for example, equation~(3.18) in \cite{billingsley2013}). Combining uniform integrability with the convergence in distribution in \eqref{main1} in  Theorem~\ref{main} proves the convergence of moments in \eqref{main2} (see, for example, Theorem~3.5 in \cite{billingsley2013}). \end{proof}

The next lemma allows us to prove Theorem~\ref{main} with $p_{0}=p$.

\begin{lemma}\label{lemmap0}
Let $\{X_{n}\}_{n\ge1}$ be any sequence of non-negative random variables and let $\{a_{n}\}_{n\ge1}$ be any sequence of real numbers. If 
\begin{align}\label{assumpcd}
a_{n}X_{n}\to_{\textup{dist}}\textup{Exponential}(1)\quad\text{as }n\to\infty,
\end{align}
and $\lim_{n\to\infty}a_{n}/b_{n}=1$, then 
\begin{align*}
b_{n}X_{n}\to_{\textup{dist}}\textup{Exponential}(1)\quad\text{as }n\to\infty.
\end{align*}
\end{lemma}

\begin{proof}[Proof of Lemma~\ref{lemmap0}]
Fix $x\ge0$ and let $\eps>0$. Since $F(x):=1-e^{-x}$ is continuous, there exists $\eta>0$ so that $|F(x)-F(y)|\le\eps$ for all $y\in\R$ such that $|x-y|\le\eta x$. Since $\lim_{n\to\infty}a_{n}/b_{n}=1$, there exists $N_{1}\ge1$ such that $1-\eta\le a_{n}/b_{n}\le1+\eta$ for all $n\ge N_{1}$. Let $F_{n}(z):=\P(X_{n}\le z)$ for $z\in\R$. By \eqref{assumpcd}, there exists $N_{2}\ge1$ so that $|F_{n}((1\pm\eta)x/a_{n})-F((1\pm\eta)x)|\le\eps$ for all $n\ge N_{2}$.

Since $F_{n}$ is nondecreasing, we have that for $n\ge\max\{N_{1},N_{2}\}$,
\begin{align*}
\P(b_{n}X_{n}\le x)-F(x)
&=F_{n}((a_{n}/b_{n})x/a_{n})-F(x)\\
&\le F_{n}((1+\eta)x/a_{n})-F(x)\\
&\le\eps+F((1+\eta)x)-F(x)\le2\eps.
\end{align*}
Again using that $F_{n}$ is nondecreasing, we similarly have that for $n\ge\max\{N_{1},N_{2}\}$,
\begin{align*}
\P(b_{n}X_{n}\le x)-F(x)
&=F_{n}((a_{n}/b_{n})x/a_{n})-F(x)\\
&\ge F_{n}((1-\eta)x/a_{n})-F(x)\\
&\ge-\eps+F((1-\eta)x)-F(x)\ge-2\eps.
\end{align*}
Since $\eps>0$ is arbitrary, the proof is complete.
\end{proof}

The following lemma gives the intuitive result that $\sigma^{-}$ and $\sigma$ have equivalent moments for large $r$.

\begin{lemma} \label{momentsequiv}
Under the assumptions of section~\ref{mainsection}, we have
\begin{align*}
    \mathbb{E}[(\sigma^-)^m] \sim \mathbb{E}[\sigma^m]\quad \text{as }r\to\infty.
\end{align*}
\end{lemma}

\begin{proof}[Proof of Lemma~\ref{momentsequiv}]

By definition of conditional probability,
\begin{align} \label{eq:p3a}
    S_{\sigma^-}(t) = \frac{\mathbb{P}(t< \sigma < \tau)}{\mathbb{P}(\sigma < \tau)}.
\end{align}
Since $\sigma$ and $\tau$ are independent, the denominator of \eqref{eq:p3a} is
\begin{align} \label{eq:p3b}
    \mathbb{P}(\sigma < \tau) = \mathbb{E}[S_{\tau}(\sigma)] = \int_0^{\infty} S_{\tau}(s)\,\dd F_{\sigma}(s).
\end{align}
Similarly, the numerator of \eqref{eq:p3a} is
\begin{align} \label{eq:p3c}
    \mathbb{P}(t< \sigma < \tau) = \mathbb{E}[S_{\tau}(\sigma)\mathbbm{1}_{\sigma>t}] = \int_t^{\infty} S_{\tau}(s)\,\dd F_{\sigma}(s)
\end{align}
Substituting the expressions in (\ref{eq:p3b}) and (\ref{eq:p3c}) into (\ref{eq:p3a}),
\begin{align} \label{eq:p3d}
    S_{\sigma^-}(t) = \frac{\int_t^{\infty} S_{\tau}(s)\,\dd F_{\sigma}(s)}{\int_0^{\infty} S_{\tau}(s)\,\dd F_{\sigma}(s)}.
\end{align}
By \eqref{eq:p3d},
\begin{align} \label{eq:p4a}
    \int_0^{\infty} mt^{m-1}S_{\sigma^-}(t)\,\text{d}t = \frac{\int_0^{\infty} \int_t^{\infty} mt^{m-1} S_{\tau}(s)\,\dd F_{\sigma}(s)\,\text{d}t}{\int_0^{\infty} S_{\tau}(s)\,\dd F_{\sigma}(s)}.
\end{align}
Given the non-negativity of the integrand in the numerator of (\ref{eq:p4a}), Tonelli's theorem implies
\begin{align} \label{eq:p4b}
    \int_0^{\infty} \int_t^{\infty} mt^{m-1} S_{\tau}(s)\,\dd F_{\sigma}(s)\,\text{d}t 
    &= \int_0^{\infty} s^m S_{\tau}(s)re^{-rs}\,\dd F_{\sigma}(s).
\end{align}
Substituting (\ref{eq:p4b}) into (\ref{eq:p4a}) yields
\begin{align} \label{eq:p4c}
    \int_0^{\infty} mt^{m-1}S_{\sigma^-}(t)\,\text{d}t 
    = \frac{\int_0^{\infty} t^m S_{\tau}(t)\,\dd F_{\sigma}(t)}{\int_0^{\infty} S_{\tau}(t)\,\dd F_{\sigma}(t)}
    =\frac{\int_0^{\infty} t^m S_{\tau}(t)\,\dd F_{\sigma}(t)}{1-p}.
\end{align}
Changing variables $t=x^{1/m}$ yields
\begin{align} \label{eq:p4d}
    \mathbb{E}[(\sigma^-)^m] = \int_0^{\infty} \mathbb{P}(\sigma^- > x^{1/m})\,\text{d}x = \int_0^{\infty} mt^{m-1}S_{\sigma^-}(t)\,\text{d}t,
\end{align}
and thus \eqref{eq:p4c} implies
\begin{align} \label{eq:p4f}
    \mathbb{E}[(\sigma^-)^m] &= \frac{\int_0^{\infty} t^m S_{\tau}(t)\,\dd F_{\sigma}(t)}{1-p}.
\end{align}

Let $\eps>0$. Since $\lim_{r\to\infty}p=0$ by Lemma~\ref{pvanish}, we can take $r$ sufficiently large so that $1/(1-p)\le1+\eps$. Hence, \eqref{eq:p4f} implies that for sufficiently large $r$,
\begin{align*}
\mathbb{E}[(\sigma^-)^m]
\le(1+\eps)\int_0^{\infty} t^m S_{\tau}(t)\,\dd F_{\sigma}(t)
\le(1+\eps)\int_0^{\infty} t^m \,\dd F_{\sigma}(t)
=(1+\eps)\E[\sigma^{m}],
\end{align*}
since $S_{\tau}(t)\le1$ for all $t$. Therefore,
\begin{align*}
\limsup_{r\to\infty} \frac{\mathbb{E}[(\sigma^-)^m]}{\E[\sigma^{m}]}
    \leq 1+\eps.
\end{align*}

Now we consider the limit infimum. The denominator of (\ref{eq:p4f}) is clearly bounded above by unity. By right-continuity of $S_{\tau}(t)$ and the assumption in \eqref{Ftau0} that $F_{\tau}(0)=0$, there exists $\eta>0$ such that for $t\in(0,\eta)$, $S_{\tau}(t) \geq 1 - \varepsilon.$ Thus,
\begin{align} \label{eq:p4g}
    \int_0^{\infty} t^m S_{\tau}(t)\,\dd F_{\sigma}(t)
    &\geq \int_0^{\eta} t^m S_{\tau}(t)\,\dd F_{\sigma}(t) \nonumber\\
    &\geq (1-\varepsilon) \int_0^{\eta} t^m \,\dd F_{\sigma}(t) \nonumber\\
    &\ge (1-\eps)\E[\sigma^{m}\mathbbm{1}_{\sigma<\eta}]
    = (1-\eps)r^{-m}\E[Y^{m}\mathbbm{1}_{Y<r\eta}],
\end{align}
where $\mathbbm{1}_A$ denotes the indicator function on an event $A$. Hence, using $\E[\sigma^{m}]=r^{-m}\E[Y^{m}]$ and the monotone convergence theorem we obtain
\begin{align*} 
    \liminf_{r\to\infty} \frac{\mathbb{E}[(\sigma^-)^m]}{\E[\sigma^{m}]} \geq 1-\varepsilon.
\end{align*}
Since $\varepsilon\in(0,1)$ is arbitrary, the proof is complete.
\end{proof}

We ultimately use convergence of moment generating functions to conclude Theorem~\ref{main}. Hence, we use Lemma \ref{pvanish} below to obtain the desired large $r$ behavior of the moment generating functions for $\sigma^-$ and $\tau^-$.

\begin{lemma} \label{pvanish}
Under the assumptions of section~\ref{mainsection}, we have
\begin{align*}
    \lim_{r\to\infty} p = 0.
\end{align*}
\end{lemma}

\begin{proof}[Proof of Lemma \ref{pvanish}] Since $\sigma$ is exponentially distributed with rate $r>0$,
\begin{align*}
    p \coloneqq  \mathbb{P}(\tau<\sigma)
    = \int_0^{\infty} S_{\sigma}(t)\,\dd F_{\tau}(t)
    = \int_0^{\infty} F_{\tau}(t)\,\dd F_{\sigma}(t),
\end{align*}
where the second equality follows from integration by parts. Let $\eps\in(0,1)$. Since $F_{\tau}(t)$ is right-continuous and since $F_{\tau}(0)=0$ by assumption~\eqref{Ftau0}, there exists $\eta>0$ such that $F_{\tau}(t)<\varepsilon$ for $t\in(0,\eta)$. Hence,
\begin{align*}
    0 < p \leq \varepsilon\int_0^{\eta} \,\dd F_{\sigma}(t) + \int_{\eta}^{\infty} \,\dd F_{\sigma}(t)
    \le\eps+\P(\sigma\ge\eta)
    =\eps+\P(Y\ge r\eta).
\end{align*}
Since $\lim_{r\to\infty}\P(Y\ge r\eta)=0$ by \eqref{Yunitmean} and since $\varepsilon\in(0,1)$ is arbitrary, the proof is complete.
\end{proof}

The remaining lemmas characterize the large $r$ behavior of individual terms in the moment generating function of $prT$. Altogether, these results make possible the proof of Theorem~\ref{main}.

\begin{lemma} \label{lem6}
Under the assumptions of section \ref{mainsection}, we have that for all $z\in(-\delta,\delta)$ with $\delta>0$ as in \eqref{ymgf},
\begin{align*}
    \mathbb{E}[e^{zpr \sigma^-}] = 1 + zp + o(p)\quad \text{as }r\to\infty.
\end{align*}
\end{lemma}

\begin{proof}[Proof of Lemma \ref{lem6}]
It follows immediately from \eqref{eq:p4f} that
\begin{align*}
\E[(\sigma^{-})^{m}]\le\frac{\E[\sigma^{m}]}{1-p},
\end{align*}
and since $\lim_{r\to\infty}p=0$ by Lemma~\ref{pvanish}, we have that
\begin{align*}
\E[(\sigma^{-})^{m}]
\le 2\E[\sigma^{m}]\quad\text{for sufficiently large $r$}.
\end{align*}
Now, expanding the exponential function and assuming $z\in(-\delta,\delta)$ for $\delta$ as in \eqref{ymgf} yields
\begin{align}\label{rety}
\mathbb{E}[e^{zpr \sigma^-}]
= \sum_{k=0}^{\infty} \frac{(zpr)^k}{k!} \mathbb{E}[(\sigma^-)^k]
= \sum_{k=0}^{\infty} \frac{(zp)^k}{k!} \frac{\mathbb{E}[(\sigma^-)^k]}{\E[\sigma^{k}]}\E[Y^{k}],
\end{align}
where we have used that $\sigma=Y/r$. To see the validity of exchanging the expectation with the sum in the first equality in \eqref{rety}, first note the bound
\begin{align}\label{bound22}
\begin{split}
\Big|\frac{(zpr)^k}{k!} \E[(\sigma^-)^k]\Big|
&\le\frac{(|z|r)^k}{k!} 2\E[\sigma^{k}]\\
&=\frac{|z|^k}{k!} 2\E[Y^{k}]\quad\text{for sufficiently large $r$}.
\end{split}
\end{align}
Next, Tonelli's theorem implies that
\begin{align}\label{ton22}
\sum_{k=0}^{\infty}\frac{|z|^{k}}{k!}\E[Y^{k}]
=\E\bigg[\sum_{k=0}^{\infty}\frac{|z|^{k}}{k!}Y^{k}\bigg]
=\E[e^{|z|Y}]<\infty,
\end{align}
by the assumption in \eqref{ymgf} that $Y$ has a finite moment generating function in a neighborhood of the origin. Hence, the bound in \eqref{bound22} and the finiteness of the sum in \eqref{ton22} allows us to use the dominated convergence theorem to verify the first equality in \eqref{rety}.

It follows from \eqref{rety} and the assumption in \eqref{Yunitmean} that  $\E[Y]=1$ that 
\begin{align} \label{eq:p6d}
\frac{1}{p}\Big(1 + zp - \mathbb{E}[e^{zpr \sigma^-}]\Big)
&= \frac{1}{p}\Big(1 + zp - \sum_{k=0}^{\infty} \frac{(zp)^k}{k!} \frac{\mathbb{E}[(\sigma^-)^k]}{\E[\sigma^{k}]}\E[Y^{k}]\Big) \nonumber\\
&= z\Big(1 - \frac{\mathbb{E}[\sigma^-]}{\mathbb{E}[\sigma]}\Big) - \sum_{k=2}^{\infty} \frac{p^{k-1}z^{k}}{k!} \frac{\mathbb{E}[(\sigma^-)^k]}{\E[\sigma^{k}]}\E[Y^{k}].
\end{align}
Since $\E[\sigma^{-}]\sim\E[\sigma]$ as $r\to\infty$ by Lemma~\ref{momentsequiv}, it remains to show that the sum in \eqref{eq:p6d} vanishes as $r\to\infty$. Note first that the terms in this sum vanish as $r\to\infty$ since $\lim_{r\to\infty}p=0$ by Lemma~\ref{pvanish}. It thus remains to justify exchanging the large $r$ limit with the sum. First, observe that for sufficiently large $r$ we have the bound
\begin{align*}
\Big|\frac{p^{k-1}z^{k}}{k!} \frac{\mathbb{E}[(\sigma^-)^k]}{\E[\sigma^{k}]}\E[Y^{k}]\Big|
\le\frac{|z|^{k}}{k!} 2\E[Y^{k}],
\end{align*}
and thus again using \eqref{ton22} and the dominated convergence theorem allows us to conclude
\begin{align*}
\lim_{r\to\infty}\sum_{k=2}^{\infty} \frac{p^{k-1}z^{k}}{k!} \frac{\mathbb{E}[(\sigma^-)^k]}{\E[\sigma^{k}]}\E[Y^{k}]
=\sum_{k=2}^{\infty} \lim_{r\to\infty}\frac{p^{k-1}z^{k}}{k!} \frac{\mathbb{E}[(\sigma^-)^k]}{\E[\sigma^{k}]}\E[Y^{k}]
=0,
\end{align*}
which completes the proof.
\end{proof}

\begin{lemma} \label{lem7}
Under the assumptions of section \ref{mainsection}, we have that for all $z\in(-\delta,\delta)$ with $\delta>0$ as in \eqref{ymgf},
\begin{align*}
    \lim_{r\to\infty} \mathbb{E}[e^{zpr \tau^-}] = 1.
\end{align*}
\end{lemma}

\begin{proof}[Proof of Lemma \ref{lem7}]

Since $Y$ has a finite moment generating function as in \eqref{ymgf}, its survival probability, $S_{Y}(y)=\P(Y>y)$, vanishes no slower than exponentially,
\begin{align*}
S_{Y}(y)
\le Ke^{-{{\delta}} y}\quad \text{for all }y\in\R,
\end{align*}
where $K=\E[e^{\delta Y}]<\infty$. To see this, let $\delta>0$ be as in \eqref{ymgf} and observe that Chebyshev's inequality (see, for example, Theorem 1.6.4 in \cite{durrett2019}) implies that for any $y\in\R$,
\begin{align*}
S_{Y}(y)
=\P(Y>y)
=\P(e^{\delta Y}>e^{\delta y})
\le e^{-\delta y}\E[e^{\delta Y}]<\infty.
\end{align*}
Hence, by definition of conditional probability,
\begin{align*}
    S_{\tau^-}(t)
    = \frac{\mathbb{P}(t < \tau < \sigma)}{\mathbb{P}(\tau < \sigma)}
    \leq \text{min}\Big\{1, \frac{\mathbb{P}(t < \sigma)}{\mathbb{P}(\tau < \sigma)}\Big\}
    &= \text{min}\Big\{1,\frac{1}{p}S_{Y}(rt)\Big\}\\
    &\le\begin{cases}
    \frac{K}{p}e^{-r{{\delta}} t} & \text{if }t\ge C,\\
    1 & \text{if }t<C,
    \end{cases}
\end{align*}
where
\begin{align}\label{Cdef}
C=(r\delta)^{-1}\ln(K/p).
\end{align}
Now, if $E$ is exponentially distributed with unit mean, then
\begin{align*}
\P(C+(r\delta)^{-1}E>t)
    &=\begin{cases}
    \frac{K}{p}e^{-r\delta t} & \text{if }t\ge C,\\
    1 & \text{if }t<C,
    \end{cases}
\end{align*}
and therefore,
\begin{align*}
S_{\tau^{-}}(t)
\le \P(C+(r\delta)^{-1}E>t)\quad\text{for all }t\in\R.
\end{align*}

Hence, for $z\in[0,{{\delta}})$, the non-decreasing nature of $f(x)=e^{zx}$ yields
\begin{align} \label{eq:p7c1}
1
\leq \mathbb{E}[e^{zpr \tau^-}]
\leq \mathbb{E}[e^{zpr(C+(r\delta)^{-1}E)}]
=\frac{e^{zprC}}{1-zp/{{\delta}}}.
\end{align}
For $z\in(-{{\delta}},0)$, the non-increasing nature of $f(x)=e^{-|z|x}$ yields
\begin{align} \label{eq:p7c2}
1
\ge \mathbb{E}[e^{-|z|pr \tau^-}]
\ge \mathbb{E}[e^{-|z|pr(C+(r\delta)^{-1}E)}]
=\frac{e^{zprC}}{1-zp/{{\delta}}}.
\end{align}
Taking $r\to\infty$ in \eqref{eq:p7c1} and \eqref{eq:p7c2} and using the definition of $C$ in \eqref{Cdef} and the fact that $\lim_{r\to\infty}p=0$ by Lemma~\ref{pvanish} completes the proof.
\end{proof}

\subsection{Proofs of results in section~\ref{charsection}}

The proofs of Proposition \ref{proplap} and Theorems \ref{thmlog}-\ref{thm4} below capture the utility of the results in Theorem~\ref{main} given increasingly strong assumptions on the cumulative distribution function of the search process without resetting.

\begin{proposition} \label{proplap}
Assume $C,d>0$ and $b\in\mathbb{R}$. If $\eta>0$, then
\begin{align} \label{eq:proplap}
    \int_0^{\eta} e^{-rt} t^b e^{-C/t^d}\,\textup{d}t \sim {\mu} r^{\beta - 1} \exp(-\gamma r^{d/(d+1)})\quad \textup{as }r\to\infty,
\end{align}
where
\begin{align*}
    {\mu} = \sqrt{\frac{2\pi (Cd)^{\frac{2b+1}{d+1}}}{d+1}}, \quad \beta = \frac{d-2b}{2d+2}, \quad \gamma = \frac{d+1}{d^{d/(d+1)}}C^{\frac{1}{d+1}}.
\end{align*}
\end{proposition}

\begin{proof}[Proof of Proposition \ref{proplap}] One can verify that the exponential factor in the integrand of \eqref{eq:proplap} is maximized at
\begin{align*}
    t^* \coloneqq \Big(\frac{Cd}{r}\Big)^{1/(d+1)}.
\end{align*}
Let $r>0$ be sufficiently large such that $t^*\in(0,\eta)$. The change of variables $s = t/t^*$ transforms the integral in \eqref{eq:proplap} into
\begin{align} \label{intLPM}
    \int_0^{\eta} &e^{-rt}t^be^{-C/t^d}\,\text{d}t \nonumber\\
    &= \Big(\frac{Cd}{r}\Big)^{\frac{b+1}{d+1}} \int_0^{\eta/t^*} s^b \exp\Big[ r^{\frac{d}{d+1}} \Big( -s(Cd)^{\frac{1}{d+1}} - Cs^{-d} (Cd)^{\frac{-d}{d+1}} \Big) \Big]\,\text{d}s
\end{align}
so that $s=1$ corresponds to the maximum of the exponential. From here, we can apply Laplace's method to \eqref{intLPM},
\begin{align} \label{pf_proplap}
    \int_0^{\eta/t^{*}} f(s) e^{x\phi(s)}\,\text{d}s \sim \frac{\sqrt{2\pi}f(1)e^{x\phi(1)}}{\sqrt{-x\phi''(1)}},\quad \text{as }x\coloneqq r^{d/(d+1)}\to\infty,
\end{align}
where
\begin{align}
    f(s) = s^b,\quad \phi(s) = -s(Cd)^{\frac{1}{d+1}} - Cs^{-d} (Cd)^{\frac{-d}{d+1}}.
\end{align}
Simplifying the right-hand side of \eqref{pf_proplap} completes the proof.
\end{proof}

\begin{proof}[Proof of Theorem~\ref{thmlognormal}]
The proof of Theorem~\ref{thmlognormal} follows from the proof of Theorem~\ref{thmlog} below with $d=1$.
\end{proof}

\begin{proof}[Proof of Theorem~\ref{thmlog}]
Let ${{\eps}}>0$. By the assumption in \eqref{lim_thmlog}, there exists a $\eta>0$ such that 
\begin{align} \label{Ftau_thmlog}
e^{-(C+{{\eps}})/t^d}
\le    F_{\tau}(t) 
\le e^{-(C-{{\eps}})/t^d}\quad \text{for all } t\in(0,\eta).
\end{align}
For any $u,v\in\R\cup\{\infty\}$, define $I_{u,v} \coloneqq \int_u^v e^{-rt}\,\text{d}F_{\tau}(t)$ so that $p=I_{0,\eta}+I_{\eta,\infty}$. Using integration by parts and the upper bound in \eqref{Ftau_thmlog} yields
\begin{align}
    I_{0,\eta} &= e^{-r\eta}F_{\tau}(\eta) + \int_0^{\eta} re^{-rt} F_{\tau}(t)\,\text{d}t\nonumber\\
    &\leq e^{-r\eta}F_{\tau}(\eta) + \int_0^{\eta} re^{-rt}e^{-(C-{{\eps}})/t^d}\,\text{d}t. \label{pf_thmlog}
\end{align}
Similarly, the lower bound in \eqref{Ftau_thmlog} yields
\begin{align*}
\int_0^{\eta} re^{-rt}e^{-(C+{{\eps}})/t^d}\,\text{d}t
\le I_{0,\eta}.
\end{align*}
Now, Proposition \ref{proplap} implies that 
\begin{align*}
\int_0^{\eta} re^{-rt}e^{-(C\pm{{\eps}})/t^d}\,\text{d}t
\sim {\mu}_{\pm\eps}r^{\beta}\exp(-\gamma_{\pm\eps}r^{d/(d+1)})\quad\text{as }r\to\infty,
\end{align*}
where 
\begin{align*}
 {\mu}_{\pm\eps} = \sqrt{\frac{2\pi ((C\pm\eps)d)^{\frac{2b+1}{d+1}}}{d+1}}, \quad \beta = \frac{d}{2d+2}, \quad \gamma_{\pm\eps} = \frac{d+1}{d^{d/(d+1)}}(C\pm\eps)^{\frac{1}{d+1}}.
\end{align*}
It is straightforward to check that $I_{\eta,\infty}=O(re^{-r\eta})$ as $r\to\infty$, and thus
\begin{align*}
\lim_{r\to\infty}\frac{e^{-r\eta}F_{\tau}(\eta)}{\int_0^{\eta} re^{-rt}e^{-(C\pm{{\eps}})/t^d}\,\text{d}t}
=\lim_{r\to\infty}\frac{I_{\eta,\infty}}{\int_0^{\eta} re^{-rt}e^{-(C\pm{{\eps}})/t^d}\,\text{d}t}
=0.
\end{align*}
Therefore,
\begin{align*}
&\limsup_{r\to\infty}r^{-d/(d+1)}\ln p
=\limsup_{r\to\infty}r^{-d/(d+1)}\ln(I_{0,\eta}+I_{\eta,\infty})\\
&\quad\le\limsup_{r\to\infty}r^{-d/(d+1)}\ln(e^{-r\eta}F_{\tau}(\eta)+\int_0^{\eta} re^{-rt}e^{-(C-{{\eps}})/t^d}\,\text{d}t
+I_{\eta,\infty})
=-\gamma_{-\eps}.
\end{align*}
The analogous calculation on using the lower bound in \eqref{Ftau_thmlog} finally yields
\begin{align*}
-\gamma_{+\eps}
\le\liminf_{r\to\infty}r^{-d/(d+1)}\ln p
\le\limsup_{r\to\infty}r^{-d/(d+1)}\ln p
\le-\gamma_{-\eps}.
\end{align*}
Since $\eps>0$ is arbitrary, we have proven \eqref{thmlogpdp}, and \eqref{thmlogpdt} follows from \eqref{main2}.
\end{proof}

\begin{proof}[Proof of Theorem~\ref{thmlinearnormal}]
The proof of Theorem~\ref{thmlinearnormal} follows from the proof of Theorem~\ref{thm4} below with $d=1$.
\end{proof}

\begin{proof}[Proof of Theorem~\ref{thm4}]
Let ${{\eps}}\in(0,1)$. By the assumption in \eqref{thm4F}, there exists $\eta>0$ such that
\begin{align}
    (1-{{\eps}})At^be^{-C/t^d} \leq F_{\tau}(t) \leq (1+{{\eps}})At^be^{-C/t^d},\quad \text{for all } t\in(0,\eta).
\end{align}
For any $u,v\in\R\cup\{\infty\}$, define $I_{u,v} \coloneqq \int_u^v e^{-rt}\,\text{d}F_{\tau}(t)$ so that $p$ in \eqref{eq:p} is $p=I_{0,\infty}=I_{0,\eta}+I_{\eta,\infty}$.  Using integration by parts and the assumption in \eqref{Ftau0}, we bound $I_{0,\eta}$ from above,
\begin{align} \label{I0del}
    I_{0,\eta} &= e^{-r\eta}F_{\tau}(\eta) + \int_0^{\eta} re^{-rt}F_{\tau}(t)\,\text{d}t\nonumber \\
    &\leq e^{-r\eta}F_{\tau}(\eta) + (1+{{\eps}})Ar \int_0^{\eta} e^{-rt}t^be^{-C/t^d}\,\text{d}t.
\end{align}
By Proposition \ref{proplap} and the fact that $I_{\eta,\infty}$ and the boundary term in \eqref{I0del} are $O(re^{-\eta r})$ as $r\to\infty$, we conclude
\begin{align}
    \limsup_{r\to\infty} \frac{I_{0,\infty}}{{\mu} r^{\beta}\exp( -\gamma r^{d/(d+1)})} \leq 1+{{\eps}}.
\end{align}
The analogous argument for the lower bound yields
\begin{align}
    1-{{\eps}} \leq \liminf_{r\to\infty} \frac{I_{0,\infty}}{{\mu} r^{\beta}\exp( -\gamma r^{d/(d+1)})}.
\end{align}
Since ${{\eps}}\in(0,1)$ is arbitrary, the proof is complete.
\end{proof}

\subsection{Calculations from examples}

We now give the details of the calculations for the examples in section~\ref{examples}.

\subsubsection{Diffusive search}

In section~\ref{diff}, we considered a searcher that diffuses with diffusivity $D>0$ in $d\ge1$ spatial dimensions and starts at (and is reset to) a position that is distance $L>0$ from the target. In particular, we considered the following three scenarios: (i) $d=1$ and the target is a single point, (ii) $d=3$ and the target is a sphere of radius $a>0$, and (iii) $d=3$ and the target is the exterior of a sphere centered at the starting and resetting position (i.e.\ the FPT is the first time the searcher escapes a sphere of radius $L>0$). The FPT distributions without resetting for these three scenarios are given respectively by \cite{carslaw1959}
\begin{align*}
F_{\tau}(t)
&= \text{erfc}\bigg(\sqrt{\frac{L^2}{4Dt}}\bigg),\\
F_{\tau}(t)
&=\Big(\frac{a}{a+L}\Big)\text{erfc}\bigg(\sqrt{\frac{L^2}{4Dt}}\bigg),\\
F_{\tau}(t)
&= \sqrt{\frac{4L^2}{\pi Dt}}\sum_{j=0}^{\infty} \exp\Bigg( \frac{-(j+1/2)^2L^2}{Dt} \Bigg),
\end{align*}
where $\text{erfc}(z)=(2/\sqrt{\pi})\int_{z}^{\infty}e^{-u^{2}}\,\dd u$ denotes the complementary error function. Expanding these expressions as $t\to0^{+}$ (and using that $\text{erfc}(z)\sim e^{-z^{2}}/\sqrt{\pi z^{2}}$ as $z\to\infty$) and applying Theorem~\ref{thmlinearnormal} then yields the formulas for $p_{0}$ in \eqref{p0sdiff}.

To obtain the exact distribution and moments used, we find the Laplace transform $\widetilde{S}(s)$ of the survival probability,
\begin{align*}
S(t):=\P(\tau>t)=1-F_{\tau}(t),
\end{align*}
for these three scenarios. Using that the Laplace transform of $\textup{erfc}(\sqrt{c/t})$ is $\int_{0}^{\infty}e^{-st}\textup{erfc}(\sqrt{c/t})\,\dd t=e^{-2\sqrt{cs}}/s$ for $c>0$, we obtain that the Laplace transforms of the survival probability for the first two scenarios are
\begin{align}\label{tilde12}
\widetilde{S}(s)
=\frac{e^{sL^{2}/D}}{s},\quad
\widetilde{S}(s)
=\Big(\frac{a}{a+L}\Big)\frac{e^{sL^{2}/D}}{s}.
\end{align}
For the third scenario of escape from a sphere, we first recall that the survival probability conditioned on the starting radius $X(0)=x\in[0,L]$ satisfies the backward Fokker-Planck equation,
\begin{align}\label{bfp}
\partial_{t}S
=(D\partial_{xx}S+(2/x)\partial_{x}S),\quad x\in(0,L),
\end{align}
with boundary condition $S=0$ at $x=L$ and unit initial condition. Laplace transforming \eqref{bfp} yields a linear ordinary differential equation which can be solved to obtain the Laplace transform of the survival probability conditioned on starting at the center of the sphere,
\begin{align}\label{tilde3}
\widetilde{S}(s)
=\frac{1}{s}\Big(1-(\sqrt{sL^{2}/D})\text{csch}(\sqrt{sL^{2}/D})\Big),
\end{align}
where $\text{csch}(z)=2/(e^{z}-e^{-z})$.

Having obtained the Laplace transform of the survival probability of the FPT with no resetting, the distribution of the FPT with resetting at rate $r>0$ (i.e.\ $T$ in \eqref{eq:T}) can be computed by numerical Laplace inversion with the general relation \cite{evans2020},
\begin{align}\label{gen}
\widetilde{S}_{(r)}(s)
:=\int_{0}^{\infty}e^{-st}\P(T>t)\,\dd t
=\frac{\widetilde{S}(r+s)}{1-r\widetilde{S}(r+s)},\quad s\ge0.
\end{align}
Further, the moments of $T$ can be obtained using \eqref{gen} and the general relation 
\begin{align*}
\E[T^{m}]
=m(-1)^{m-1}\frac{\dd^{m-1}}{\dd s^{m-1}}\widetilde{S}_{(r)}(s)\Big|_{s=0}.
\end{align*}

\subsubsection{Diffusive search with uniform initial conditions}

For the example considered in section~\ref{unif} of diffusive search in the interval $[0,L]$ with uniform initial and resetting conditions, recall that the survival probability conditioned on starting at $x\in[0,L]$ satisfies the backward Fokker-Planck equation,
\begin{align}\label{bfp2}
\partial_{t}S
=D\partial_{xx}S,\quad x\in(0,L),
\end{align}
with absorbing boundary conditions $S=0$ at $x\in\{0,L\}$ and unit initial condition. Laplace transforming \eqref{bfp2}, dividing by $L$ and integrating from $x=0$ to $x=L$ yields the Laplace transform of the survival probability conditioned on a uniformly distributed initial position,
\begin{align}\label{Stildeunif}
\widetilde{S}(s)
&=\frac{1}{s}\Big(1-\frac{ \tanh (\sqrt{sL^{2}/(4D)})}{\sqrt{sL^{2}/(4D)}}\Big).
\end{align}
Taking $s\to\infty$ and using Tauberian theorems (see, for example, Theorem~3 in section 5 of chapter 8 of \cite{feller-vol-2}) yields the short-time behavior of $F_{\tau}(t)$ in \eqref{shortunif}, which then yields the asymptotic probability $p_{0}$ in \eqref{p0unifexample} via Proposition~\ref{pfeller}. Further, the FPT distribution under stochastic resetting is then obtained via numerical Laplace inversion of \eqref{Stildeunif} using \eqref{gen}.

For the example of diffusive exit from a sphere with uniform initial and resetting conditions, we Laplace transform \eqref{bfp}, solve the resulting linear ordinary differential equation analytically, multiply the solution by $(3/L^{3})x^{2}$, and integrate from $x=0$ to $x=L$ to obtain the Laplace transform of the survival probability conditioned on a uniformly distributed initial position,
\begin{align*}
\widetilde{S}(s)
=\frac{-3 \sqrt{sL^{2}D} \coth (\sqrt{sL^{2}/D})+3 D+sL^2}{(sL)^{2}}.
\end{align*}
As above, taking $s\to\infty$ and using Tauberian theorems yields the short-time behavior of $F_{\tau}(t)$,
\begin{align*}
F_{\tau}(t)
\sim\sqrt{\frac{36Dt}{\pi L^{2}}}\quad\text{as }t\to0^{+},
\end{align*}
which then yields the following asymptotic probability via Proposition~\ref{pfeller},
\begin{align*}
p\sim
p_{0}=3\sqrt{D/(rL^{2})}\quad\text{as }r\to\infty.
\end{align*}
Further, the FPT distribution under stochastic resetting is then obtained via numerical Laplace inversion of \eqref{Stildeunif} using \eqref{gen}.

\subsubsection{Search on a discrete network}

For the example considered in section~\ref{ctmc} of search on a discrete network, recall that the dynamics of the continuous-time Markov chain $X$ are described by the infinitesimal generator matrix $Q\in\R^{|I|\times|I|}$. (The entry $Q(i,j)$ in the $i$th column and $j$th row of $Q$ is the jump rate from state $i$ to state $j$ if $i\neq j$ and the diagonal entries $Q(i,i)$ are chosen so that $Q$ has zero row sums \cite{norris1998}.)

To compute the survival probability of the FPT without resetting, $S(t):=\P(\tau>t)$, let the target be a single node, $I_{\t}=i_{\t}\in I$, and let $\widehat{Q}$ denote the matrix obtained by deleting the row and column in $Q$ corresponding to $i_{\t}$. Similarly, for an initial distribution $\rho$, let $\widehat{\rho}$ denote the vector obtained by deleting the entry in $\rho$ corresponding to $i_{\t}$. Then, $S(t)$ is given by the sum of the entries in the vector $e^{W_{0}t}\widehat{\rho}$, where $W_{0}$ is the transpose of $\widehat{Q}$ and $e^{W_{0}t}$ is the matrix exponential \cite{lawley2019imp}. In particular,
\begin{align}\label{dot}
S(t)
 =\mathbf{1}\cdot e^{W_{0}t}\widehat{\rho},
\end{align}
where $\cdot$ denotes the dot product and $\mathbf{1}\in\R^{|I|-1}$ is the vector of all ones. Taking the Laplace transform of \eqref{dot} then yields
\begin{align}\label{Stildectmc}
\widetilde{S}(s)
=\mathbf{1}\cdot(s\,\text{id}-W_{0})^{-1}\widehat{\rho},
\end{align}
where $\text{id}$ denotes the identity matrix. The FPT distribution under stochastic resetting is then obtained via numerical Laplace inversion of \eqref{Stildectmc} using \eqref{gen}.

The particular continuous-time Markov chains (i.e.\ choices of $Q$) used in Figure~\ref{figunif}b are created following a method used in \cite{lawley2020networks}. Specifically, we first construct a graph by randomly connecting $|I|\gg1$ vertices by approximately $5|I|$ edges. We then assign jump rates to each directed edge independently according to a uniform distribution. That is, $Q(i,i)\le0$ are chosen so that $Q$ has zero row sums and the off-diagonal entries are
\begin{align*}
Q(i,j)
=\begin{cases}
U_{i,j} & \text{if there is a directed edge from $i$ to $j$},\\
0 & \text{otherwise},
\end{cases}
\end{align*}
where $\{U_{i,j}\}_{i,j\in I}$ are independent uniform random variables on $[0,1]$.

\subsubsection{Run and tumble}

For the example considered in section~\ref{rtp} of a run and tumble search, the asymptotic probability $p_{0}$ in \eqref{p0rtp} is computed by using \eqref{eq:p}, integrating by parts, and changing variables to obtain
\begin{align*}
p
=\int_{0}^{\infty}e^{-rt}\,\dd F_{\tau}(t)
=re^{-rt_{0}}\int_{0}^{\infty}e^{-rt}F_{\tau}(t_{0}+t)\,\dd t,
\end{align*}
where $t_{0}=L/V>0$ is the smallest possible value of $\tau$ (since the searcher starts distance $L>0$ from the target and moves at a finite speed $V$). Next, observe that
\begin{align*}
F_{\tau}(t_{0}+t)
=\P(\tau\le t_{0}+t)
&=\P(\tau=t_{0})+\P(t_{0}<\tau<t_{0}+t)+\P(\tau=t_{0}+t)\\
&=\P(\tau=t_{0})+\P(t_{0}<\tau<t_{0}+t),
\end{align*}
since $\P(\tau<t_{0})=\P(\tau=t_{0}+t)=0$ for $t>0$ and $\P(\tau=t_{0})=\frac{1}{2}e^{-\lambda t_{0}}$ is the probability that the searcher starts in the $-V<0$ velocity state and does not switch states before time $t_{0}$. Hence,
\begin{align}\label{gr}
p
=\tfrac{1}{2}e^{-(r+\lambda)t_{0}}
+re^{-rt_{0}}\int_{0}^{\infty}e^{-rt}\P(t_{0}<\tau<t_{0}+t)\,\dd t.
\end{align}
Next, it was shown in section~4.1 of \cite{lawley2021pdmp} that 
\begin{align*}
\P(t_{0}<\tau<t_{0}+t)
=\beta t+O(t^{2})\quad\text{as }t\to0^{+},
\end{align*}
where $\beta=\lambda  e^{-\frac{\lambda  L}{V}} (\lambda  L+V)/(4V)$. Therefore applying Proposition~\ref{pfeller} to the integral in \eqref{gr} yields \eqref{p0rtp}.

To compute the exact FPT distribution, recall that the survival probability $S_{\pm}(t;x)$ conditioned on starting at $x\ge0$ in the state moving to the right or left (denoted by $+$ or $-$) satisfies the backward Fokker-Planck equation,
\begin{align}
\begin{split}\label{bfp3}
\partial_{t}S_{+}
&=V\partial_{x}S_{+}+\lambda(S_{-}-S_{+}),\\
\partial_{t}S_{-}
&=-V\partial_{x}S_{-}+\lambda(S_{+}-S_{-}),
\end{split}
\end{align}
with boundary condition $S_{-}=0$ at $x=0$, far-field condition $S_{+}=1$ as $x\to\infty$, and unit initial conditions. Laplace transforming \eqref{bfp3} yields a pair of linear ordinary differential equations which can be solved to obtain the Laplace transform of the survival probability conditioned on starting at $x=L$ with probability $1/2$ of being in either the $+$ or $-$ state,
\begin{align}\label{Stildertp}
\widetilde{S}(s)
=\frac{1}{2 \lambda  s}e^{-\frac{L \sqrt{s (2 \lambda +s)}}{V}} \Big(2 \lambda  \big(e^{\frac{L \sqrt{s (2 \lambda +s)}}{V}}-1\big)+\sqrt{s (2 \lambda +s)}-s\Big).
\end{align}
The FPT distribution under stochastic resetting is then obtained via numerical Laplace inversion of \eqref{Stildertp} using \eqref{gen}.

\subsubsection{Subdiffusive search}

For the examples considered in section~\ref{sub} of subdiffusive search, the exact FPT distribution with resetting is obtained via numerical Laplace inversion of \eqref{tilde12} and \eqref{tilde3} using \eqref{gen} and \eqref{relaalpha}.

\subsubsection{Non-exponential resetting}

For the examples considered in section~\ref{nonexp}, we have $F_{\tau}(t)= \text{erfc}\Big(\sqrt{\frac{L^2}{4Dt}}\Big)$. For sharp reset (see \eqref{nonexp1}), the probability of a successful search is $p=F_{\tau}(1/r)$. For uniform reset (see \eqref{nonexp2}), the probability of a successful search is
\begin{align*}
p
=\int_0^{\infty} S_{\sigma}(t) \,\text{d}F_{\tau}(t)
&=\int_0^{\infty} F_{\tau}(t) \,\text{d}F_{\sigma}(t)\\
&=\int_0^{2/r} \text{erfc}\bigg(\sqrt{\frac{L^2}{4Dt}}\bigg) \frac{r}{2}\,\dd t\\
&=\big(1+rL^2 /(4D)\big) \text{erfc}(\sqrt{rL^{2}/(8D)})-\sqrt{\frac{rL^{2}}{2\pi D}}  e^{-rL^2/(8 D)},
\end{align*}
where the second equality follows from integration by parts. 
Similarly, for gamma reset (see \eqref{nonexp3}), the probability of a successful search is
\begin{align*}
p
=\int_0^{\infty} F_{\tau}(t) \,\text{d}F_{\sigma}(t)
&=\int_0^{\infty} \text{erfc}\bigg(\sqrt{\frac{L^2}{4Dt}}\bigg) 4 r^2 t e^{-2 r t}\,\dd t\\
&=e^{-\sqrt{2rL^{2}/D}} (1+\sqrt{rL^{2}/(2D)}).
\end{align*}


\section*{Acknowledgements}
This material is based upon work supported by the National Science Foundation Graduate Research Fellowship Program under Grant No.\,2139322. SDL was supported by the National Science Foundation (Grant Nos.\ CAREER DMS-1944574 and DMS-1814832).

\bibliography{library.bib}
\bibliographystyle{unsrt}

\end{document}